\documentclass[11pt,a4paper]{article}
\usepackage{epsfig}
\usepackage{graphics}
\usepackage{subfigure}

\usepackage{epstopdf}

\usepackage{soul}
\usepackage{ulem}
\usepackage[english]{babel}
\usepackage{mdwlist}
\usepackage{graphicx}
\usepackage{hyperref} \hypersetup{backref, colorlinks=true, linkcolor=magenta,citecolor=blue}
\usepackage{appendix}
\usepackage{dsfont}
\usepackage{tikz}
\usepackage{amsmath}
\usepackage{amsbsy}
\usepackage{amsfonts}
\usepackage{amssymb}
\usepackage{amscd}
\usepackage{amsthm}
\usepackage{empheq}
\usepackage{placeins} 

\usepackage{lscape} 

\usepackage{diagbox} 
\usepackage{array} 

\usepackage{bm}
\usepackage{xspace}

\usepackage{subfigure}

\makeatletter
\def\XS{\xspace}
\DeclareMathAlphabet{\mathb}{OML}{cmm}{b}{it}
\def\sbm#1{\ensuremath{\mathb{#1}}}                
\def\sbmm#1{\ensuremath{\boldsymbol{#1}}}          
\def\scu#1{\ensuremath{\mathcal{#1\XS}}}           

\def\Ab{{\sbm{A}}\XS}  \def\ab{{\sbm{a}}\XS}

  \def\hb{{\sbm{h}}\XS}

\def\Mb{{\sbm{M}}\XS}

  \def\pb{{\sbm{p}}\XS}
  \def\qb{{\sbm{q}}\XS}

  \def\xb{{\sbm{x}}\XS}
  \def\yb{{\sbm{y}}\XS}
  \def\zb{{\sbm{z}}\XS}

\def\Dc{{\scu{D}}\XS}   
\def\Ec{{\scu{E}}\XS}

\def\Ic{{\scu{I}}\XS}

\def\Nc{{\scu{N}}\XS}

\def\Cbb{{\sbl{C}}\XS}

\def\Nbb{{\sbl{N}}\XS}

\def\Rbb{{\sbl{R}}\XS}

\def\pib         {{\sbmm{\pi}}\XS}

\def\Card#1{\mathop{Card}\left(#1\right)}

\def\sgn{{\mathrm{sign}}}							

\newsavebox{\fminibox}
\newlength{\fminilength}

  \def\+{^\dagger}

\def\nequiv{\not\kern-.05em\equiv}
\def\egal{\kern-.5em=\kern-.5em}        
\def\propt{\kern-.2em\propto\kern-.2em} 

\def\argmin{\mathop{\mathrm{arg\,min}}} 

\def\intdouble{\int\kern-0.3em\int}
\def\inttriple{\int\kern-0.3em\int\kern-0.3em\int}

\def\rond#1{\overset{\kern-0.33em~_\circ}{#1}}
\def\rondit[#1]#2{\overset{\kern#1~_\circ}{#2}}

\newtheorem{definition}{Definition}[section]

\def\qed{\ifmmode\hbox{\hfill\sqb}\else{\ifhmode\unskip\fi%
\nobreak\hfil
\penalty50\hskip1em\null\nobreak\hfil$\blacksquare$
\parfillskip=0pt\finalhyphendemerits=0\endgraf}\fi}

\def\argmin{\mathop{\mathrm{arg\,min}}} 
\RequirePackage{amsmath}
\RequirePackage{xspace}

\def\XS{\xspace}
\DeclareMathAlphabet{\mathb}{OML}{cmm}{b}{it}
\def\sbm#1{\ensuremath{\mathb{#1}}}                
\def\sbmm#1{\ensuremath{\boldsymbol{#1}}}          
\def\scu#1{\ensuremath{\mathcal{#1\XS}}}           
\def\sbl#1{\ensuremath{\mathds{#1}}}              

\def\Nc{{\scu{N}}\XS}

\def\Ab{{\sbm{A}}\XS}  \def\ab{{\sbm{a}}\XS}

  \def\hb{{\sbm{h}}\XS}

\def\Mb{{\sbm{M}}\XS}

  \def\pb{{\sbm{p}}\XS}
  \def\qb{{\sbm{q}}\XS}

  \def\xb{{\sbm{x}}\XS}
  \def\yb{{\sbm{y}}\XS}
  \def\zb{{\sbm{z}}\XS}

\def\Dc{{\scu{D}}\XS}   
\def\Ec{{\scu{E}}\XS}

\def\Ic{{\scu{I}}\XS}

\def\Nc{{\scu{N}}\XS}

\def\MAP{^{\kern1pt{\rm MAP}\kern-1pt}}

\usepackage{color,soul}

\def\Cbb{{\sbl{C}}\XS}

\def\Nbb{{\sbl{N}}\XS}

\def\Rbb{{\sbl{R}}\XS}

\def\pib         {{\sbmm{\pi}}\XS}

\def\PM{\kern0pt^{\textrm{{\scriptsize PM}}}\kern0pt}
\def\MMAP{\kern1pt^{\textrm{{\tiny MMAP}}}\kern-1pt} 

\def\rem#1{}                    

 \def\btabu{\begin{tabular}}             \def\etabu{\end{tabular}}
 
 \usepackage{mathtools}
\newcommand{\defeqt}{\vcentcolon=}

 \def\nbloc{m}
 \def\scalprodE#1#2{\left\langle #1 , #2 \right\rangle_{E^* \times E}}
  \def\scalprodF#1#2{\left\langle #1 , #2 \right\rangle_{F^* \times F}}
 \def\ns{{n_s}}
\makeatother

\newtheorem{thmchapter}{Theorem}[section]

\newtheorem{prop}[thmchapter]{Proposition}

\newtheorem{lemme}[thmchapter]{Lemma}

\newtheorem{hyp}[thmchapter]{Assumption}

\usepackage{dsfont}
\usepackage{algorithm} 
\usepackage{algpseudocode}
\usepackage{tikz}
\newcommand*\circled[1]{\tikz[baseline=(char.base)]{%
            \node[shape=circle,draw,inner sep=1pt] (char) {#1};}}

\title{An algorithm for variable density sampling with block-constrained acquisition}

\author{Claire Boyer$^{(1)}$, Pierre Weiss$^{(2)}$, and J\'er\'emie Bigot$^{(3)}$ \vspace{0.2cm}  \\
\\
$^{(1)}$ Institut de Math\'ematiques de Toulouse, Universit\'e de Toulouse, France\\
{\small {claire.boyer}@math.univ-toulouse.fr} \\
$^{(2)}$ Institut des Technologies Avanc\'ees du Vivant, Toulouse, France\\
{\small {pierre.armand.weiss}@gmail.com} \\
$^{(3)}$ DMIA, Institut Sup\'erieur de l'A\'eronautique et de l'Espace, Toulouse, France\\
{\small {jeremie.bigot}@isae.fr} 
 \vspace{0.2cm}}

\begin{document}
\maketitle

\begin{abstract}

Reducing acquisition time is of fundamental importance in various imaging modalities.
The concept of variable density sampling provides an appealing framework to address this issue. It was justified recently from a theoretical point of view in the compressed sensing (CS) literature. Unfortunately, the sampling schemes suggested by current CS theories may not be relevant since they do not take the acquisition constraints into account (for example, continuity of the acquisition trajectory in Magnetic Resonance Imaging - MRI). 
In this paper, we propose a numerical method to perform variable density sampling with block constraints. Our main contribution is to propose a new way to draw the blocks in order to mimic CS strategies based on isolated measurements. The basic idea is to minimize a tailored dissimilarity measure between a probability distribution defined on the set of isolated measurements and a probability distribution defined on a set of blocks of measurements. This problem turns out to be convex and solvable in high dimension.
Our second contribution is to define an efficient minimization algorithm based on Nesterov's accelerated gradient descent in metric spaces. We study carefully the choice of the metrics and of the prox function. We show that the optimal choice may depend on the type of blocks under consideration. Finally, we show that we can obtain better MRI reconstruction results using our sampling schemes than standard strategies such as equiangularly distributed radial lines.
\end{abstract}


\section{Introduction}

Compressive Sensing (CS) is a recently developed sampling theory that provides theoretical conditions to ensure the exact recovery of signals from a few number of linear measurements (below the Nyquist rate). CS is based on the assumption that the signal to reconstruct can be represented by a few number of atoms in a certain basis. We say that the signal $\xb \in \Cbb^n$ is $s$-sparse if
$$ \left\| \xb \right\|_{\ell^0}  \leq s,
$$
where $\left\| \cdot \right\|_{\ell^0}$ denotes the $\ell_0$ pseudo-norm, counting the number of non-zero entries of $\xb$. Original CS theorems \cite{donoho2006compressed,candes2006robust,candes2011probabilistic} assert that a sparse signal $\xb$ can be faithfully reconstructed via $\ell_1$-minimization:
\begin{align}
\label{pbMin1}
\min_{\zb \in \Cbb^n} \left\| \zb \right\|_{\ell^1} \qquad \text{such that } \qquad \Ab_\Omega \zb = \yb,
\end{align}
where $\Ab_\Omega\in \Cbb^{p\times n}$ ($p\leq n$) is a sensing matrix,
$\yb = \Ab_\Omega \xb \in \Cbb^p$ represents the vector of linear projections, and $\| \zb \|_{\ell^1} = \sum_{i=1}^n | z_i |$ for all $\zb = \left( z_1 , \hdots , z_n \right) \in \Cbb^n$.  More precisely CS results state that $p=O(s\ln(n))$ measurements are enough to guarantee exact reconstruction provided that $\Ab_\Omega$ satisfies some incoherence property. 

One way to construct $\Ab_\Omega$ is by randomly extracting rows from a full sensing matrix $\Ab \in \Cbb^{n \times n}$  that can be written as
\begin{align}
\label{def:fullSensingMatrix}
\Ab = \begin{pmatrix}
\ab_1^* \\
\vdots\\
\ab_n^*
\end{pmatrix},
\end{align} 
where $\ab^*_i$ denotes the $i$-th row of  $\Ab$. In the context of Magnetic Resonance Imaging (MRI) for instance, the full sensing matrix $\Ab$ consists in the composition of a Fourier transform with an inverse wavelet transform. This choice is due to the fact that the acquisition is done in the Fourier domain, while the images to be reconstructed are assumed to be sparse in the wavelet domain. In this setting, a fundamental issue  is constructing $\Ab_\Omega$ by extracting  appropriate rows from the  full sensing matrix  $\Ab$. A theoretically founded approach to build $\Ab_\Omega$  (i.e.\ constructing of sampling schemes) consists in randomly extracting rows from $\Ab$ according to a given density. This approach requires to define a  discrete probability distribution $\pb = (\pb_{i})_{1 \leq i \leq n}$ on the set of integers $\{1,\ldots,n\}$  that represents the indexes of the rows of $\Ab$. We call this procedure variable density sampling. This term appeared in the early MRI paper \cite{spielman1995magnetic}. It was recently given a mathematical definition in \cite{chauffertSIAM}.
One possibility to construct $\pb$ is to choose its i-th component $\pb_{i}$ to be proportional to $\|\ab^*_i \|_{\ell^\infty}^2$ (see \cite{rauhut2010compressive,puy2011variable,bigot2013analysis,chauffert2013variable}) i.e.\
\begin{equation}
\pb_{i} = \frac{\|\ab^*_i \|_{\ell^\infty}^2}{\sum_{k=1}^{n} \|\ab^*_k \|_{\ell^\infty}^2}, \; i=1,\ldots,n. \label{eq:popt}
\end{equation}
In the MRI setting, another strategy ensuring good reconstruction is to choose $\pb$ according to a polynomial radial distribution \cite{krahmer2012beyond} in the  so-called \textit{k-space} i.e.\ the 2D Fourier plane where low frequencies are centered. Other strategies are possible. For example, \cite{adcock2013breaking} propose a multilevel uniformly random subsampling approach. 

All these strategies lead to sampling schemes that are made of a few but isolated measurements, see e.g.\ Figure \ref{fig:blocks} (a).
However, in many applications, the number of measurements is not of primary  importance relative to the path the sensor must take to collect the measurements. For instance, in MRI, sampling is done in the Fourier domain along continuous and smooth curves \cite{wright1997MRI,lustig2008fast}. Another example of the need to sample  continuous  trajectories can be found in mobile robots monitoring where  robots have to spatially sample their environment under kinematic and energy consumption constraints \cite{hummel2011mission}.

This paper focuses on the acquisition of linear measurements in applications where the physics of the sensing device allows to sample a signal from pre-defined blocks of measurements. We define a block of measurements as an arbitrary set of isolated measurements, that could be contiguous in the Fourier plane for instance.  As an illustrative example (that will be used throughout the paper), one may consider sampling patterns generated by randomly drawing a set of straight lines in the Fourier plane or \textit{k-space} as displayed in Figure \ref{fig:blocks}(b). This kind of sampling patterns is particularly relevant in the case of MRI acquisition with echo planar sampling strategies, see e.g.\ \cite{lustig2008compressed}. 

Acquiring data by blocks of measurements raises the issue of designing appropriate sampling schemes. In this paper, we propose to randomly extract blocks of measurements that are made of several rows from the full sensing matrix $\Ab$. The main question investigated is how to choose an appropriate probability distribution from which blocks of measurements will be drawn. A first step in this direction \cite{bigot2013analysis,polak2012performance} was recently proposed. In \cite{bigot2013analysis}, we have  derived a theoretical probability distribution in the case of blocks of measurements to design a sensing matrix $\Ab_{\Omega}$ that guarantees an exact reconstruction of $s$-sparse signals with high probability. Unfortunately, the probability distributions proposed in  \cite{bigot2013analysis} and \cite{polak2012performance} are difficult to compute numerically and seem suboptimal in practice. 

In this paper, we propose an alternative strategy which is based on the numerical resolution of an optimization problem. Our main idea is to construct a probability distribution $\pib$ on a dictionary of blocks. The blocks are drawn independently at random according to this distribution. We propose to choose $\pib$ in such a way that the resulting sampling patterns are similar to those based on isolated measurements, such as the ones proposed in the CS literature. For this purpose, we define a dissimilarity measure to compare a probability distribution $\pib$  on a dictionary of blocks  and a target probability distribution $\pb$ defined on a set of isolated measurements. Then, we propose to choose an appropriate distribution $\pib\left[\pb \right]$ by minimizing its dissimilarity with a distribution $\pb$ on isolated measurements that is known to lead to good sensing matrices. 

This paper is organized as follows. In Section \ref{sec:notation}, we introduce the notation. In Section \ref{sec:VDS}, we describe the problem setting. Then, we construct  a dissimilarity measure between  probability distributions lying in different, but spatially related domains. We then formulate the problem of finding a probability distribution $\pib\left[ \pb \right]$ on blocks of measurements as a convex optimization problem. In Section \ref{sec:optim}, we present an original and efficient way to solve this minimization problem via a dual formulation and an algorithm based on the accelerated gradient descents in metric spaces \cite{nesterov2005smooth}.  We study carefully how the theoretical rates of convergence are affected by the choice of norms and prox-functions on the primal and dual spaces. Finally, in Section \ref{sec:num}, we propose a dictionary of blocks  that is appropriate for MRI applications. Then, we compare the quality of MRI images reconstructions using the proposed sampling schemes and those currently used in the context of MRI acquisition, demonstrating the potential of the proposed approach on real scanners.
 
\begin{figure}
\begin{center}
\btabu{@{}cc}
\includegraphics[height=5cm]{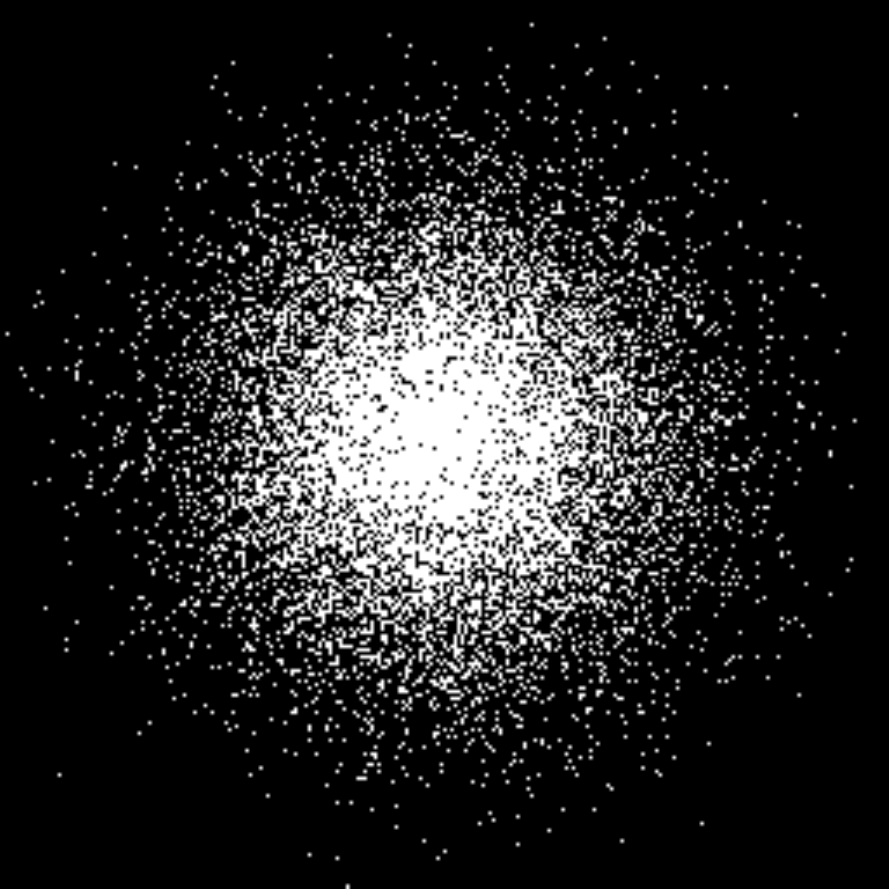} &
\includegraphics[height=5cm]{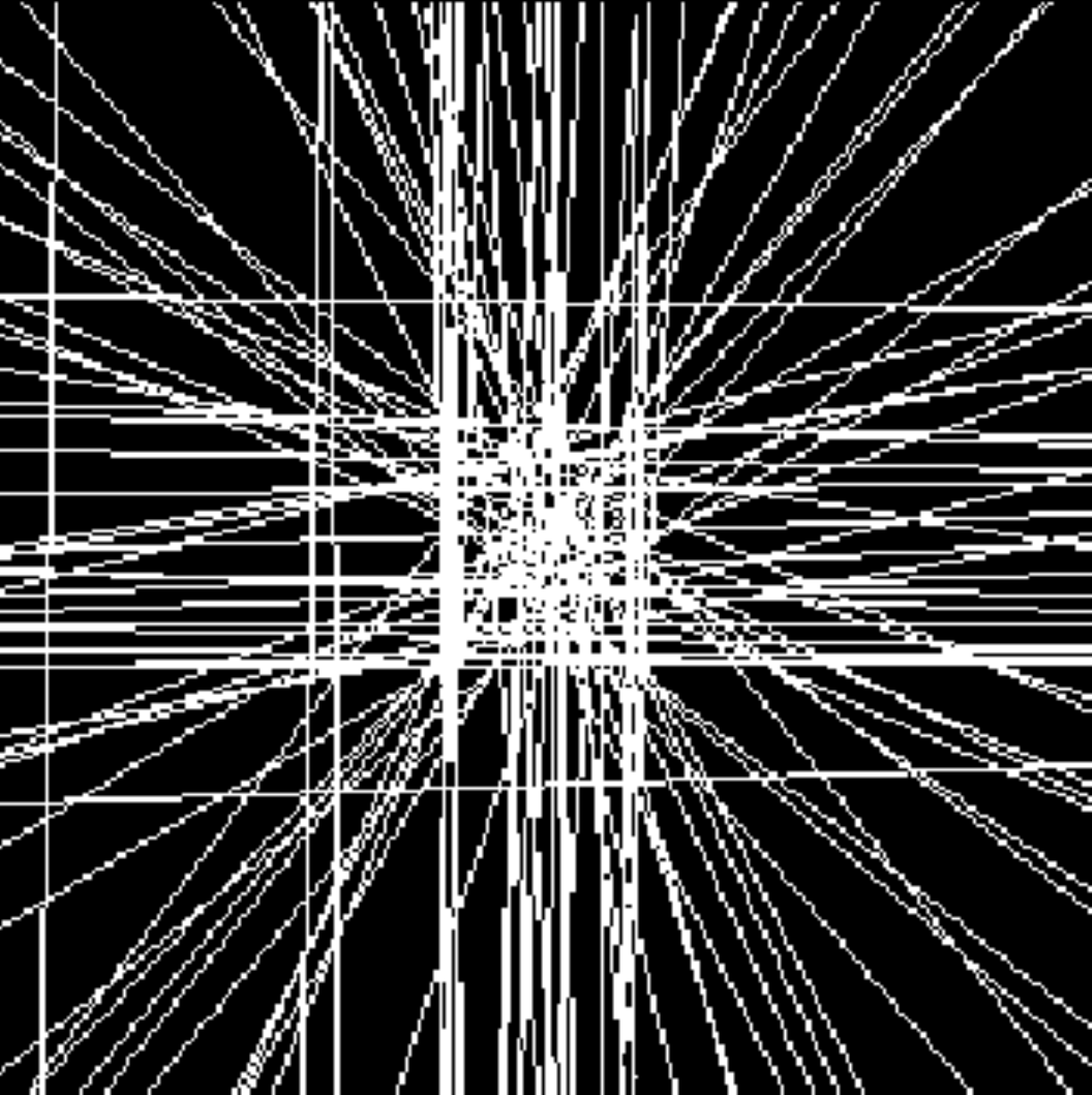} \\
{\small (a)}&{\small (b)} 
\etabu
\caption{\label{fig:blocks}{\bf An example of MRI sampling schemes in the \textit{k-space} (the 2D Fourier plane where low frequencies are centered) } 
(a): Isolated measurements drawn from a probability measure $\pb$ having a radial distribution. (b): 
Sampling scheme based on a dictionary of blocks of measurements: blocks consist of discrete lines of the same size.
}
\end{center}
\end{figure}

\section{Notation}
\label{sec:notation}
We consider $d$-dimensional signals for any $d \in \Nbb^*$, of size $n_1\times n_2 \times \hdots \times n_d = n $. 
Let $E$ and $F$ denote finite-dimensional vector spaces endowed with their respective norms $\|.\|_E$ and $\| . \|_F$. 
In the paper, we identify $E$ to $\Rbb^\nbloc$ and $F$ to $\Rbb^n$. 
We denote by $E^*$ and $F^*$, respectively the dual spaces of $E$ and $F$. For $s\in E^*$ and $x\in E$ we denote by
$\left\langle s, x \right\rangle_{E^*\times E}$ the value of $s$ at $x$. The notation $\left\langle \cdot , \cdot \right\rangle$ will denote the usual inner product in a Euclidean space.
The norm of the dual space $E^*$ is defined by:
$$ \left\| s \right\|_{E^*} = \max_{ x \in E \atop \|x\|_E = 1}  \scalprodE{s}{x}. $$
Let $\Mb : E \rightarrow F^*$ denote some operator. When $M$ is linear, we denote its adjoint operator by $\Mb^*  : F \rightarrow E^*$.
The subordinate operator norm is defined by :
\begin{equation*}
 \|\Mb\|_{E\to F^*} = \sup_{\|x\|_E\leq 1} \|\Mb x\|_{F^*}
\end{equation*}
When the spaces $E^*$ and $F$ are endowed with $\ell^q$ and $\ell^p$ norms respectively, we will use the following notation for the operator norm of $\Mb^*$:
$$ \| \Mb^* \|_{F \rightarrow E^*} = \| \Mb^* \|_{p \rightarrow q}.
$$
We set $\Delta_\nbloc \subset E$ to be the simplex in $E=\Rbb^\nbloc$, and $\Delta_n \subset F$ to be the simplex in $F=\Rbb^n$.
For $\pib \in \Delta_\nbloc$ and an index $j \in \left\lbrace 1 , \hdots , \nbloc \right\rbrace$ we denote by $\pib_j$   the $j$-th component of $\pib$.

Let $g:\Rbb^n\rightarrow \Rbb\cup\{+\infty\}$ denote a closed convex function. Its Fenchel conjugate is denoted $g^*$. The relative interior of a set $X\subseteq \Rbb^n$ is denoted $\text{ri}(X)$. Finally, the normal cone to $X$ at a point $x$ on the boundary of $X$ is denoted $\Nc_X(x)$.

\section{Variable density sampling with block constraints}

\label{sec:VDS}

\subsection{Problem setting}

In this paper, we assume that the acquisition system is capable of sensing a finite set $\{y_1, \hdots , y_n\}$ of linear measurements of a signal $\xb \in \Rbb^{\ns}$ such that $$y_i = \langle \ab_i^*, \xb\rangle , \qquad\forall i = 1, \hdots , n,$$
where $\ab_i^*$ denotes the $i$-th row of the full sensing matrix $\Ab$. 
Let us define a set $\Ic= \{ I_1 ,\hdots , I_m\}$ where each $I_k \subseteq \{1,\hdots, n\}$ denotes a set of indexes. We assume that the acquisition system has physical constraints that impose
sensing simultaneously the following sets of measurements $$E_k =\{ y_i , i \in I_k \}, \qquad \forall k = 1 , \hdots, m.$$ 
In what follows, we refer to $\Ic$ as the blocks dictionary.

For example in MRI, $n = \ns$ is the number of pixels or voxels of a 2D or 3D image, and $y_i$
represents the $i$-th discrete Fourier coefficient of this image. In this setting, the sets of indexes $I_k$ may represent straight lines in the discrete Fourier domain as in Figure \ref{fig:blocks}(b). In Section \ref{subsec:dico}, we give further details on the construction of such a dictionary.

We propose to partially sense the signal using the following procedure:
\begin{enumerate}
\item Construct a discrete probability distribution $\pib \in \Delta_m$.
\item Draw  i.i.d.\ indexes  $k_1, \hdots, k_b$  from the probability distribution $\pib$ on the set $\{1 , \hdots , m\}$, with $1 \leq b \leq m$.
\item Sense randomly the signal $\xb$ by considering the random set of blocks of measurements $\left( E_{k_j} \right)_{j \in \{ 1,\hdots, b\}}$, which leads to the construction of the following sensing matrix $$ \Ab_\Omega = \left( \ab_i^* \right)_{\displaystyle i \in \cup_{j=1}^b I_{k_j}}.$$
\end{enumerate}

The main objective of this paper is to provide an algorithm to construct the discrete probability distribution $\pib$ based on the knowledge of a target discrete probability distribution $\pb \in\Delta_n$ on the set $\{ y_1 , \hdots, y_n\}$ of isolated measurements. The problem of choosing a distribution $\pb$ leading to good image reconstruction is not addressed in this
paper, since there already exist various theoretical results and heuristic strategies in the CS literature on this topic \cite{lustig2008fast,chauffert2013variable,adcock2013breaking,krahmer2012beyond}.  

\subsection{A variational formulation}

In order to define $\pib$, we propose to minimize a dissimilarity measure between $\pib \in \Delta_m$ and $\pb \in \Delta_n$.
The difficulty lies in the fact that these two probability distributions belong to different spaces. We propose to
construct a dissimilarity measure $\Dc( \pib, \pb , \Ic )$ that depends on the blocks dictionary $\Ic$. This dissimilarity measure will be minimized over $\pib \in \Delta_m$ using numerical algorithms with $m$ being relatively large (typically $10^{4} \leq m \leq 10^{10}$). Therefore, it must have appropriate properties
such as convexity, for the problem to be solvable in an efficient way.

\subsubsection*{Mapping the $m$-dimensional simplex to the $n$-dimensional one}

In order to define a reasonable dissimilarity measure, we propose to construct an operator $\Mb$ that maps a
probability distribution $\pib \in \Delta_m$ to some $\pb' \in \Delta_n$:

\begin{align*}
\Mb : \qquad & E    \longrightarrow F^* \\
& \pib  \longmapsto \pb',
\end{align*}
where for $i \in \left\lbrace 1 , \hdots, n \right\rbrace$, 
\begin{align}
\label{eq:M}
\pb'_i = \displaystyle \frac{\sum_{k=1}^\nbloc \pib_k {\mathds{1}_{i \in I_k}}}{\sum_{j = 1}^n \sum_{k'=1}^\nbloc \pib_{k'} {\mathds{1}_{j \in I_{k'}}}},
\end{align}
where $\mathds{1}_{i \in I_k}$ is equal to 1 if $i \in I_k$, 0 otherwise.
The $i$-th element of $\pb'$ represents the probability to draw the $i$-th measurement $y_i$ by drawing blocks of
measurements according to the probability distribution $\pib$. The operator $\Mb$ satisfies the following property by construction : 
$$\Mb \Delta_m \subseteq \Delta_n.$$

\subsubsection*{A sufficient condition for the mapping $\Mb$ to be a linear operator}

Note that the operator $\Mb$ is generally non linear, due to the denominator in \eqref{eq:M}. This is usually
an important drawback for the design of numerical algorithms involving the operator $\Mb$. However, if the sets $\left(I_k\right)_{k \in \{1,\hdots , m\}}$ all have  the same cardinality (or length) equal to $\ell$, the denominator in \eqref{eq:M} is equal to $\ell$. In this case, $\Mb$ becomes a linear operator. In this paper, we will focus on this setting, which is rich enough for many practical applications: 
\begin{hyp}
\label{hyp:card}
For $k \in \{ 1, \hdots, m \}$,  $ \Card{ I_k} = \ell$, where $\ell$ is some positive integer.
\end{hyp}

Let us provide two important results for the sequel. 
\begin{prop} 
\label{prop:subset}
For $\ell >1$,
$\Mb \Delta_m \subsetneq \Delta_n$, i.e. $\Mb \Delta_m$ is a strict subset of $\Delta_n$.
\end{prop}
\begin{proof}
By definition of the convex envelope, 
$\Mb \Delta_m = \text{conv} \left( \left\lbrace \Mb_{:,i} , i \in \{ 1,\hdots , m \} \right\rbrace \right)$, where $\Mb_{:,i}$ denotes the $i$-th column of $\Mb$. 
For $\ell>1$, $\left\lbrace \Mb_{:,i} , i \in \{ 1,\hdots , m \} \right\rbrace$ is a subset of $\Delta_n$ that does not contain the extreme points of the simplex.
\end{proof}

In practice, Proposition \ref{prop:subset} means that it is impossible to reach exactly an arbitrary distribution $\pb \in \Delta_n$, except for the trivial case of isolated measurements.
\begin{prop}
\label{prop:normMFE}
Suppose that Assumption \ref{hyp:card} holds, then for $p \in \left[ 1 , \infty \right]$,
$$ \| \Mb^* \|_{p \rightarrow \infty} = \ell^{-\frac{1}{p}}.
$$
\end{prop}
\begin{proof}
Under Assumpiton \ref{hyp:card}, all the columns of $\Mb$ have only $\ell$ non-zero coefficients equal to $1/\ell$. With $\|\cdot \|_F = \|\cdot \|_{\ell^p}$, we can thus derive that
\begin{align*}
\| \Mb^* \|_{p \rightarrow \infty} &= \max_{\|x\|_p=1} \| \Mb^* x\|_{\ell^\infty} = \max_{1\leq i \leq m}\max_{\|x\|_p=1}  \left\langle \Mb_{:,i} , x \right\rangle \\
&= \max_{1\leq i \leq m} \| \Mb_{:,i} \|_{F^*}  = \max_{1\leq i \leq m} \| \Mb_{:,i} \|_{q} \\
& =\ell^{-\frac{1}{p}},
\end{align*}
where $\Mb_{:,i}$ denotes the $i$-th column of $\Mb$, and $q$ is the conjugate of $p$ satisfying $1/p+1/q=1$.
\end{proof}

\subsubsection*{Measuring the dissimilarity between $\pib$ and $\pb$ through the operator $M$}

Now that we have introduced the mapping $\Mb$, we propose to define a dissimilarity measure between $\pib \in \Delta_m$ and $\pb \in \Delta_n$. To do so, we propose to compare $\Mb \pib$ and $\pb$ that are both vectors belonging to the simplex $\Delta_n$. Owing to Proposition \ref{prop:subset}, it is hopeless to find some $\tilde{\pib} \in \Delta_m$ satisfying $\Mb \tilde{\pib} = \pb$ for an arbitrary target density $\pb$. Therefore, we can only expect to get an approximate solution by minimizing a dissimilarity measure $\Dc(\Mb\pib,\pb)$. For obvious numerical reasons, $\Dc$ should be convex in $\pib$. Among statistical distances, the most natural ones are the total variation distance, Kullback-Leibler of more generally f-divergences. Among this family, total variation presents the interest of having a dual of bounded support. We will exploit this property to design efficient numerical algorithms in Section \ref{sec:optim}. In the sequel, we will thus use $\Dc(\Mb\pib,\pb)=\| \Mb \pib - \pb \|_{\ell^1}$ to compare the distributions $\Mb \pib$ and $\pb$.

\subsubsection*{Entropic regularization}

In applications such as MRI,  the number $m$ of columns of $\Mb$ is larger than the number $n$ of its rows. Therefore, $\text{Ker}(\Mb)\neq \emptyset$ and there exist multiple  $\pib \in \Delta_m$ with the same dissimilarity measure $\Dc(\Mb \pib , \pb)$. In this case, we propose to take among all these solutions, the one minimizing  the neg-entropy $\Ec$ defined by
\begin{align}
\label{def:entropy}
\Ec : \pib \in \Delta_m \longmapsto \sum_{j=1}^m \pib_j \log(\pib_j),
\end{align}
with the convention that $0 \log(0) = 0$. We recall  that the entropy $\Ec(\pib)$ is proportional to the Kullback-Leibler divergence between $\pib$ and the uniform distribution $\pib^{c}$ in $\Delta_m$ (i.e.\ such that $\pib^{c}_{j} = \frac{1}{m}$ for all $j$). Therefore, among all the solutions minimizing $\Dc(\Mb \pib , \pb)$, choosing the distribution $\pib(\pb)$ minimizing $\Ec(\pib)$ gives priority to entropic solutions, i.e. probability distributions which maximize the covering of the sampling space if we proceed to several drawings of blocks of measurements. 
 Therefore, we can finally write the following regularized problem  defined by
\begin{align}
\label{pb:finalGoal}
\tag{PP}
\min_{\pib \in \Delta_m}  F_{\alpha}(\pib), 
\end{align}
where
$$
 F_{\alpha}(\pib) = \| \Mb \pib - \pb \|_{\ell^1} + \alpha \Ec(\pib),
$$
for some regularization parameter $\alpha > 0$. Adding the neg-entropy has the effect of spreading out the probability distribution $\pib$, which is a desirable property. Moreover, the neg-entropy is strongly convex on the simplex $\Delta_m$. This feature is of primary importance for the numerical resolution of the above optimization problem. Note that an appropriate choice of the regularization parameter $\alpha$ is also important, but this issue will not be addressed in this paper.

\subsubsection*{A toy example}

To illustrate the interest of Problem \eqref{pb:finalGoal}, we design a simple example.
Consider a $3\times 3$ image. Define  the target distribution $\pb$ as a dirac on the central pixel (numbered 5 in Figure \ref{fig:toyexample}). Consider a blocks dictionary composed of horizontal and vertical lines. In that setting, the operator $\Mb$ is given by
\begin{align*}
\Mb = \frac{1}{3}\begin{pmatrix}
1&0&0&0&1&0&0 \\
1&0&0&0&0&1&0 \\
1&0&0&0&0&0&1 \\
0&1&0&0&1&0&0 \\
0&1&0&0&0&1&0 \\
0&1&0&0&0&0&1 \\
0&0&1&0&1&0&0 \\
0&0&1&0&0&1&0 \\
0&0&1&0&0&0&1 
\end{pmatrix}.
\end{align*}
For such a matrix, there are various distributions minimizing $\left\|\Mb\pib-\pb\right\|_{\ell^1}$. For example, one can choose
$
\pib_1 = \begin{pmatrix}
0 &1 &0&0 &0 & 0
\end{pmatrix}^*$ or
$
\pib_2 = \begin{pmatrix}
0 &1/2 &0&0&1/2 & 0
\end{pmatrix}^*$.
The solution maximizing the entropy is $\pib_2$. In the case of image processing, this solution is preferable since it leads to better covering of the acquisition space. Note that, among all the $\ell^p$-norms ($1\leq p<+\infty)$, only the $\ell^1$-norm is such that $\left\|\Mb\pib_1-\pb\right\|_{\ell^1}=\left\|\Mb\pib_2-\pb\right\|_{\ell^1}$. This property is once again desirable since we want the regularizing term (and not the fidelity term) to force choosing the proper solution.
\begin{figure}
\begin{center}
\includegraphics[height=6cm]{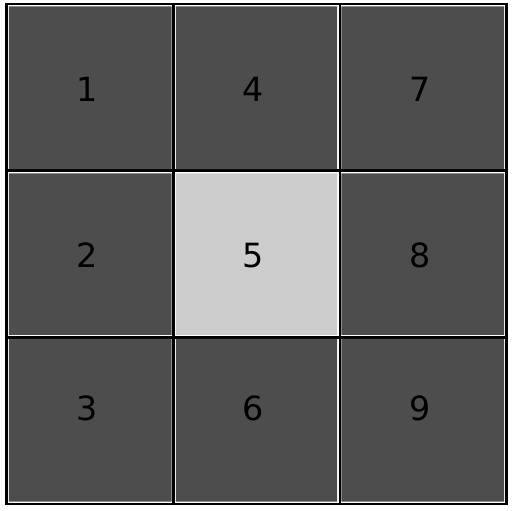}
\caption{\label{fig:toyexample}  Illustration of a target distribution concentrated on the central pixel of a $3\times 3$ images. The pixels are numbered, and this order is kept in the design of $\Mb$ and $\pib$.}
\end{center}
\end{figure}

%

\section{Optimization}
\label{sec:optim}

In this section, we propose a numerical algorithm to solve Problem \eqref{pb:finalGoal}. Note that despite being convex, this optimization problem has some particularities that make it difficult to solve. Firstly, the parameter $\pib \in \Delta_m$ lies in a very high dimensional space. In our experiments, $n$ varies between $10^4$ and $10^7$ while $m$ varies between $10^4$ and $10^{10}$. Moreover, the function $\Ec$ is differentiable but its gradient is not Lipschitz, and the total variation distance $\|\cdot\|_{\ell^1}$ is non-differentiable. 

The numerical resolution of Problem \eqref{pb:finalGoal} is thus a delicate issue. Below, we propose an efficient strategy based on the numerical optimization of the dual problem of \eqref{pb:finalGoal}, and on the use of Nesterov's ideas \cite{nesterov2005smooth}. Contrarily to most first order methods proposed recently in the literature \cite{bauschke2011convex,nesterov2007gradient,combettes2010dualization} which are based on Hilbert space formalisms, Nesterov's algorithm is stated in a (finite dimensional) normed space. We thus perform the minimization of the dual problem on a metric space, and we carefully study the optimal choice of the norms in the primal and dual spaces. We show that depending on the blocks length $\ell$, the optimal choice might well be different from the standard $\ell^2$-norm.
Such ideas stem back from (at least) \cite{chen1993convergence}, but were barely used in the domain of image processing.


\subsection{Dualization of the problem}

Our algorithm consists in solving the problem dual to \eqref{pb:finalGoal} in order to avoid the difficulties related to the non-differentiability of the $\ell^1$-norm. Proposition \ref{prop:pbDual} and \ref{prop:gradLip} state that the dual of problem \eqref{pb:finalGoal} is differentiable. We will use this feature to design an efficient first-order algorithm and use the primal-dual relationships (Proposition \ref{prop:relPrimalDual}) to retrieve the primal solution. 

\begin{prop}
\label{prop:pbDual}
Let $J_\alpha(\qb) := \scalprodF{\pb}{ \qb}  - \alpha \log \left(   \sum_{\ell=1}^\nbloc  \exp\left(  -\frac{(\Mb^*\qb)_\ell}{\alpha} \right) \right),$ for $\qb \in F$.
The dual problem to \eqref{pb:finalGoal} is:
\begin{align}
\tag{DP}
\label{TVdual}
- \min_{\qb \in B_\infty} J_\alpha(\qb) ,
\end{align}
in the sense that $\displaystyle \min_{\pib \in \Delta_m} F_{\alpha}(\pib) = \max_{\qb \in B_\infty} - J_{\alpha} (\qb)$, where $B_\infty$ is the $\ell^\infty$-ball of unit radius in $F$.
\end{prop}
\begin{proof}
The proof is available in Appendix \ref{app:pbDual}.
\end{proof}

In order to study the regularity properties of $J_\alpha$, and so the solvability of \eqref{TVdual}, we use the strong convexity of the neg-entropy $\Ec$ with respect to $\| \cdot\|_E$. First, let us recall one version of the definition of the strong convexity in Banach spaces.
\begin{definition}
We say that $f : F \rightarrow \Rbb$ is $\sigma$-strongly convex with respect to $\|\cdot \|_F$ on $F' \subset F$ if
\begin{align} 
\label{eq:strongConvexity}
\forall x, y \in F', \quad \forall t \in [0, 1], \quad f(t x + (1-t)y) \leq t f (x) + (1-t) f(y) - \frac{\sigma}{2}t(1-t) \|x-y\|_F^2.
\end{align}
We define the convexity modulus $ \sigma_f$ of $f$ as the largest positive real $\sigma$ satisfying Equation \eqref{eq:strongConvexity}.
\end{definition}

\begin{prop}
\label{prop:strongconv}
For $\| \cdot \|_E = \| \cdot \|_{\ell^p}$, $p\in [1,+\infty]$, the convexity modulus of the neg-entropy on the simplex $\Delta_m$ is $\sigma_\Ec=1$.
\end{prop}
\begin{proof}
The proof is available in Appendix \ref{app:strongconv}.
\end{proof}


\begin{prop}
\label{prop:gradLip}
The function $J_\alpha$ is convex and its gradient is Lipschitz continuous i.e.
$$ 
\| \nabla J_\alpha (\qb_1) - \nabla J_\alpha (\qb_2) \|_{F^*} \leq L_\alpha \| \qb_1 -\qb_2 \|_F \qquad \forall (\qb_1 , \qb_2) \in F^2.
$$
with constant
\begin{align}
\label{eq:Lalpha}
L_\alpha = \frac{\|\Mb^* \|_{F \rightarrow E^*}^2}{\alpha \sigma_\Ec} .
\end{align}
Moreover, $\nabla J_\alpha$ is locally Lipschitz around $\qb \in F$ with constant 
\begin{align}
\label{eq:localLipschitz}
L_\alpha(\qb) =   \frac{\|\Mb^* \|_{F \rightarrow E^*}^2}{\alpha \sigma_\Ec(\pib(\qb))},
\end{align}
where $\displaystyle \sigma_\Ec(\pib) \defeqt  \inf_{\|\hb \|_{E} = 1} \left\langle \Ec^{''}(\pib) \hb,\hb\right\rangle$ is the local convexity modulus of $\Ec$ around $\pib$, and an explicit expression for $\pib(\qb)$ is given in \eqref{eq:piFctQ}.
\end{prop}
\begin{proof}
The proof is available in Appendix \ref{app:gradLip}.
\end{proof}

Note that a standard reasoning would rather lead to $L_\alpha = \frac{\|\Mb^* \|_{2 \rightarrow 2}^2}{\alpha \sigma_\Ec}$, which is usually much larger than bound \eqref{eq:Lalpha}. Proposition \ref{prop:gradLip} implies that Problem \eqref{TVdual} is efficiently solvable by Nesterov's algorithm \cite{nesterov2005smooth}. Therefore, we will first solve the dual problem \eqref{TVdual}. Then, we use the relationships between the primal and dual solutions (as described in Proposition \ref{prop:relPrimalDual}) to finally compute   a primal solution $\pib^\star$ for Problem \eqref{pb:finalGoal}.
\begin{prop}
\label{prop:relPrimalDual}
The relationships between the primal and dual solutions
$$
\pib^\star = \argmin_{\pib \in \Delta_m} F_{\alpha}(\pib) \quad \mbox{ and }  \quad \qb^\star = \argmin_{\qb \in B_\infty}  J_{\alpha} (\qb)
$$
are given by
\begin{align} 
\label{eq:relpiQ}
\pib_j^\star = \frac{\exp\left(  -\frac{(\Mb^*\qb^\star)_j}{\alpha}    \right)}{ \sum_{k=1}^\nbloc  \exp\left(  -\frac{(\Mb^*\qb^\star)_k}{\alpha} \right)}, \qquad \forall j \in \left\lbrace 1, \hdots , \nbloc\right\rbrace.
\end{align}
Furthermore,
\begin{align}
\label{eq:interpretationVarDuale}
\sgn \left( \Mb\pib^\star -\pb\right) = \sgn  \left(\qb^\star\right).
\end{align}
\end{prop}
\begin{proof}
Equation \eqref{eq:relpiQ} is a direct consequence of \eqref{eq:piFctQ}. 
To derive the second equation \eqref{eq:interpretationVarDuale}, it suffices to write the optimality conditions of the problem $\displaystyle \max_{\qb \in B_\infty} \scalprodF{\Mb\pib^\star- \pb}{\qb}   + \alpha \Ec(\pib^\star)$. It leads to:
\begin{align*}
\Mb \pib^\star & -\pb \in \Nc_{B_\infty} ( \qb^\star) \Leftrightarrow \sgn \left( \Mb\pib^\star -\pb\right) = \sgn  \left(\qb^\star\right).
\end{align*}
\end{proof}

\subsection{Numerical optimization of the dual problem}
\label{subsec:resDual}

Now that the dual problem \eqref{TVdual} is fully characterized, we propose to solve it using Nesterov's optimal accelerated projected gradient descent \cite{nesterov2005smooth} for smooth convex optimization.

\subsubsection{The algorithm}

Nesterov's algorithm is based on the choice of a prox-function $d$ of the set $B_{\infty}$, i.e. a continuous function that is strongly convex on $B_{\infty}$ w.r.t. $\|\cdot \|_F$. Let $\sigma_d$ denote the convexity modulus of $d$, we further assume that $d(\qb_c)=0$ so that 
\begin{equation*}
d(\qb) \geq \frac{\sigma_d}{2} \| \qb - \qb_c \|_{F}^2 \qquad \forall \qb \in B_{\infty},
\end{equation*}
where $\displaystyle \qb_c  =\argmin_{\qb \in B_{\infty} } d(\qb)$. Nesterov's algorithm is described in Algorithm \ref{alg:nesterov}.

 \algrenewcommand{\alglinenumber}[1]{\scriptsize\circled{#1}}
 \begin{algorithm}
\caption{ \label{alg:nesterov}Resolution scheme for smooth optimization proposed by \cite{nesterov2005smooth}}
\begin{algorithmic}[1]
\State Initialization: choose $\qb_0 \in B_\infty$.
\For{$k = 0\hdots K$}
\vspace{0.3cm}
\State Compute $\displaystyle J_\alpha(\qb_k)$ and $\displaystyle \nabla J_\alpha(\qb_k)$

\State Find $\displaystyle \yb_k  \in \argmin_{\yb\in B_\infty} \displaystyle \left\langle \nabla J_\alpha(\qb_k), \yb - \qb_k \right\rangle + \frac{1}{2} L_\alpha \|\yb -\qb_k\|^2_F$

\State Find $\displaystyle \zb_k \in \argmin_{\qb\in B_\infty} \displaystyle \frac{L_\alpha}{\sigma_d} d(\qb) + \sum_{i=0}^k \frac{i+1}{2} \left[ J_\alpha(\qb_i) + \left\langle \nabla J_\alpha(\qb_i) , \qb - \qb_i \right\rangle \right]$

\State Set $\qb_{k+1} = \displaystyle\frac{2}{k+3} \zb_k + \frac{k+1}{k+3} \yb_k$.
\vspace{0.3cm}
\EndFor

\State Set the primal solution to $\pib_j = \displaystyle \frac{\exp\left( \displaystyle -\frac{(\Mb^*\yb_K)_j}{\alpha}    \right)}{ \sum_{k=1}^\nbloc  \exp\left( \displaystyle -\frac{(\Mb^*\yb_K)_k}{\alpha} \right)}, \qquad \forall j \in \left\lbrace 1, \hdots , \nbloc\right\rbrace.$
\end{algorithmic}
\end{algorithm}

Theorem \eqref{prop:errorBound} summarizes the theoretical guarantees of Algorithm \ref{alg:nesterov}.
\begin{thmchapter}{\cite[Theorem 2]{nesterov2005smooth}}
\label{prop:errorBound}
Algorithm \ref{alg:nesterov} ensures that
\begin{align} 
\notag
J_\alpha(\yb_k) - J_\alpha(\qb^\star) &\leq \frac{4 L_\alpha d(\qb^\star)}{\sigma_d (k+1)(k+2)}  \\
\label{eq:critCVObj}
&\leq \frac{4 \|\Mb^* \|_{F \rightarrow E^*}^2 d(\qb^\star)}{\alpha \sigma_\Ec\sigma_d (k+1)(k+2)} ,
\end{align}
where $\qb^\star$ is an optimal solution of Problem \eqref{TVdual}.
\end{thmchapter}
Since $d(\qb^\star)$ is generally unknown, we can bound \eqref{eq:critCVObj} by
\begin{align} 
\label{eq:critCVObj2}
\frac{4 \|\Mb^* \|_{F \rightarrow E^*}^2 D}{\alpha \sigma_\Ec\sigma_d (k+1)(k+2)}.
\end{align}
where $\displaystyle D = \max_{\qb \in B_{\infty}} d(\qb)$. 
Note that until now, we got theoretical guarantees in the dual space but not in the primal. 
What matters to us is rather to obtain guarantees on the primal iterates, which can be summarized by the following theorem. 
\begin{thmchapter}
\label{thm:distanceprimal}
Denote 
$$
\pib_k = \frac{\exp\left( -\frac{(\Mb^*\yb_k)}{\alpha} \right)}{\left|\left| \exp\left( -\frac{(\Mb^*\yb_k)}{\alpha} \right)\right|\right|_{\ell^1}}.
$$ 
where $\yb_k$ is defined in Algorithm \ref{alg:nesterov}. The following inequality holds:
\begin{equation*} 
 \|\pib_k - \pib^\star\|_{E}^2 \leq \frac{8 \|\Mb^* \|_{F \rightarrow E^*}^2 D}{\alpha^2 \sigma_\Ec^2\sigma_d (k+1)(k+2)}.
\end{equation*}
\end{thmchapter}
The proof is given in Appendix \ref{app:proofBoundPrimal}. It is a direct consequence of a more general result of independent interest.

\subsubsection{Choosing the prox-function and the metrics}

Algorithm \ref{alg:nesterov} depends on the choice of $\|\cdot\|_E$, $\|\cdot\|_F$ and $d$. The usual accelerated projected gradient descents consist in setting $\|\cdot\|_E=\|\cdot\|_{\ell^2}$, $\|\cdot\|_F=\|\cdot\|_{\ell^2}$ and $d(\cdot)=\frac{1}{2}\|\cdot\|^2_{\ell^2}$. However, we will see that it is possible to change the algorithm's speed of convergence by making a different choice. In this paper we concentrate on the usual $\ell^p$-norms, $p\in [1,+\infty]$. 

\paragraph{Choosing a norm on $E$:}

The following proposition shows an optimal choice for $\|\cdot\|_{E^*}$.
\begin{prop}
\label{prop:optE}
 The norm $\|\cdot\|_{E^*}$ that minimizes \eqref{eq:critCVObj2} among all $\ell^p$-norms, $p\in [1,+\infty]$ is $\|\cdot\|_{\ell^\infty}$. Note however that the minimum local Lipschitz constant $L_\alpha(\qb)$ for $\qb \in F$ might be reached for another choice of $\|\cdot\|_{E^*}$.
\end{prop}
\begin{proof}
From Proposition \ref{prop:strongconv}, we get that $\sigma_\Ec$ remains unchanged no matter how $\|\cdot\|_E$ is chosen among $\ell^p$-norms. 
The choice of $\|\cdot\|_E$ is thus driven by the minimization of $\| \Mb^* \|_{F \rightarrow E^*}$. 
From the operator norm definition, it is clear that the best choice consists in setting $\|\cdot \|_{E^*}=\|\cdot \|_{\ell^\infty}$ since the $\ell^\infty$-norm is the smallest of all $\ell^p$-norms. 
\end{proof}

According to Proposition \ref{prop:optE}, choosing $\|\cdot\|_{E^*}$ to be $\|\cdot\|_{\ell^\infty}$ leads to consider $\|\cdot\|_E$ to be $\|\cdot\|_{\ell^1}$. As shown by Proposition \ref{prop:normMFE}, it is clear that the norm $\| \Mb^* \|_{F \rightarrow E^*}$ may vary a lot with respect to $\|\cdot\|_F$ for the particular operator $\Mb$ considered in this paper. 

\paragraph{Choosing a norm on $F$ and a prox-function $d$:} by Proposition \ref{prop:optE} the norm $\|\cdot\|_F$ and the prox function $d$ should be chosen in order to minimize $\frac{\|\Mb^* \|_{F \rightarrow \infty}^2 D}{\sigma_d}$. We are unaware of a general theory to make an optimal choice despite recent progresses in that direction. The recent paper \cite{d2013optimal} proposes a systematic way of selecting $\|\cdot\|_F$ and $d$ in order to make the algorithm complexity invariant to change of coordinates for a general optimization problem. The general idea in \cite{d2013optimal} is to choose $\|\cdot \|_F$ to be the Minkowski gauge of the constraints set (of the optimization problem), and $d$ to be a strongly convex approximation of $ \frac{1}{2}\|\cdot \|_F^2$. However, this strategy is not shown to be optimal. In our setting, since the constraints set is $B_\infty$, this would lead to choose $\|\cdot\|_F = \|\cdot\|_{\ell^\infty}$. Unfortunately, there is no good strongly convex approximation of $\frac{1}{2}\|\cdot \|_{\ell^\infty}^2$. 

In this paper, we thus study the influence of $\|\cdot\|_F$ and $d$ both theoretically and experimentally, with $\|\cdot \|_F \in \left\lbrace \|\cdot \|_{\ell^1}, \|\cdot \|_{\ell^2} , \|\cdot \|_{\ell^\infty} \right\rbrace$. Propositions \ref{prop:dp},  \ref{prop:coutFinal} and \ref{prop:coutFinal2} summarize the theoretical algorithm complexity in different regimes.

\begin{prop}
\label{prop:dp}
Let $p' \in \left] 1 , 2\right]$. 
Define $d_{p'} (x) = \frac{1}{2}\|x\|_{p'}^2$. Then
\begin{itemize}
\item For $p \in [ p' , \infty ]$, $d_{p'}$ is $(p'-1)$-strongly convex w.r.t.\ $\|\cdot \|_p$.
\item For $ p \in [1 , p']$, $d_{p'}$ is $(p'-1)n^{\left(1/p'-1/p\right)}$-strongly convex w.r.t.\ $\|\cdot \|_p$.
\end{itemize} 
\end{prop}

\begin{proof}
The proof is a direct consequence of \cite[Proposition 3.6]{juditsky2008large} and of the fact that for $p'\geq p$,
$$ \|x\|_{p'} \leq \|x\|_{p} \leq n^{\left(1/p-1/p'\right)} \|x\|_{p'} .
$$
\end{proof}

\begin{prop}
\label{prop:coutFinal}
Suppose that Assumption \ref{hyp:card} holds. Set $ \|\cdot\|_F = \|\cdot\|_p$ and $d=d_{p'}$ with $p \in [1,\infty]$ and $p'\in ]1, 2]$.
For all this family of norms and prox-functions, the one minimizing the complexity bound \eqref{eq:critCVObj2} is
\begin{itemize}
\item $p'=2$ and $p\in [1 , 2]$, if $\ell^2=n$. For this choice, we get 
\begin{align}
\label{eq:boundFinal}
J_\alpha(\yb_k) - J_\alpha(\qb^\star) 
&\leq \frac{2\sqrt{n}}{\alpha (k+1)(k+2)}.
\end{align}
\item $p=p'=2$, if $\ell^2 < n $. 
For this choice, we get 
\begin{align}
\label{eq:boundFinal1}
J_\alpha(\yb_k) - J_\alpha(\qb^\star) 
&\leq \frac{2n}{\alpha \ell (k+1)(k+2)}.
\end{align}
\item $p=1$ and $p'=2$, if $\ell^2>n$. For this choice, we get 
\begin{align}
\label{eq:boundFinal2}
J_\alpha(\yb_k) - J_\alpha(\qb^\star) 
&\leq \frac{2n^{3/2}}{\alpha \ell^2 (k+1)(k+2)}.
\end{align}
\end{itemize}
\end{prop}

\begin{proof}
The result is a direct consequence of Proposition \ref{prop:dp}.
\end{proof}

Unfortunately, the bounds in \eqref{eq:boundFinal}, \eqref{eq:boundFinal1} and \eqref{eq:boundFinal2} are dimension dependent. Moreover, the optimal choice suggested by Proposition \ref{prop:coutFinal} is different from the Minkowski gauge approach suggested in \cite{d2013optimal}. Indeed, in all the cases described in Proposition \ref{prop:errorBound}, the optimal choice $\|\cdot \|_F$ differs from $\|\cdot\|_{\ell^{\infty}}$. The difficulty to apply this approach is to find a function $d\simeq 1/2 \|\cdot \|_{\ell^\infty}^2$ strongly convex w.r.t.\ $\|\cdot \|_{\ell^\infty}$. A simple choice consists in setting $d_\varepsilon = \frac{1}{2} \|\cdot \|_{\ell^\infty}^2 + \frac{\varepsilon}{2} \|\cdot\|_{\ell^2}^2$. This function is $\varepsilon$-strongly convex w.r.t.\ $\|\cdot \|_{\ell^\infty}$. We thus get the following proposition:

\begin{prop}
\label{prop:coutFinal2}
Suppose that Assumption \ref{hyp:card} holds, with $\ell = \sqrt{n}$. Set $ \|\cdot\|_F = \|\cdot\|_{\ell^\infty}$, $d_\varepsilon (\cdot) = \frac{1}{2} \|\cdot \|_{\ell^\infty}^2 + \frac{\varepsilon}{2} \|\cdot\|_{\ell^2}^2$.
\begin{align*}
J_\alpha(\yb_k) - J_\alpha(\qb^\star) 
&\leq \frac{2\left( 1/\varepsilon + n \right)}{ \alpha (k+1)(k+2)}.
\end{align*}
In particular, for $\varepsilon \propto \frac{1}{n}$, $ J_\alpha(\yb_k) - J_\alpha(\qb^\star) 
= O \left(  \frac{n}{\alpha k^2}\right) $.
\end{prop}
Note that this complexity bound is worse than that of Proposition \ref{prop:coutFinal} in the case where $\ell= \sqrt{n}$. 
In the next section, we intend to illustrate and to confirm in practice the different rates of convergence, predicted by the theoretical results in Proposition \ref{prop:coutFinal}.

\subsection{Numerical experiments on convergence}

In this section, we are willing to emphasize the improvement achieved by appropriately choosing the norms $\|.\|_E$, $\|.\|_F$, and the prox-function $d$. To do so, we run experiments on a dictionary of blocks of measurements having all the same size $\ell=256$, described in Section \ref{subsec:dico},  for 2D images of size $256 \times 256$.
At first, we choose $\|.\|_E = \|.\|_{\ell^1}$, $\|.\|_F=\|.\|_{\ell^2}$ and $d=\frac{1}{2}\|.\|_{\ell^2}^2$ and we perform Algorithm \ref{alg:nesterov} for this dictionary. In fact, this first case (the norm on E differs from the usual $\|\cdot\|_{\ell^2}$) nearly corresponds to a standard accelerated gradient descent \cite{nesterov2004introductory}.
In a second time, we set
$\|.\|_E = \|.\|_{\ell^1}$, $\|.\|_F=\|.\|_{\ell^\infty}$ $d=\frac{1}{2}\|.\|_{\ell^2}^2$. In Figure \ref{fig:CVres}, we display the decrease of the objective function in both settings. Figure \ref{fig:CVres} points out that a judicious selection of norms on $E$ and $F$ can significantly speed up the convergence: for 29 000 iterations, the standard accelerated projected gradient descent reaches a precision of $10^{-5}$ whereas Algorithm \ref{alg:nesterov} with $\|.\|_E = \|.\|_{\ell^1}$, $\|.\|_F=\|.\|_{\ell^\infty}$, i.e. a "modified" gradient descent, reaches a precision of $10^{-3}$. The conclusions for this numerical experiment appear to be faithful to what was predicted by the theory, see Proposition \ref{prop:errorBound}. For the sake of completeness, we add in Figure \ref{fig:CVres} (in green) the case where $\|.\|_E = \|.\|_{\ell^2}$, $\|.\|_F=\|.\|_{\ell^2}$ and $d=\frac{1}{2}\|.\|_{\ell^2}^2$, which is an usual choice in practice. Clearly, this is the slowest rate of convergence observed.

Finally, we perform the algorithm for $\|.\|_E = \|.\|_{\ell^1}$, $\|.\|_F=\|.\|_{\ell^2}$, and $d=\frac{1}{2}\|.\|_{\ell^2}^2$ by changing the value of $L_\alpha$. The value of $L_\alpha$ provided by Proposition \ref{prop:gradLip} is tight uniformly on $B_\infty$. However, the local Lipschitz constant of $\nabla J_\alpha$ varies rapidly inside the domain. In practice, the Lipschitz constant around the minimizer may be much smaller than $L_\alpha$ (note that $\pi^\star\in \text{ri}(\Delta_m)$ for all $\alpha>0$). In this last heuristic approach, we will decrease $L_\alpha$ by substantial factors without losing practical convergence. This result is presented in Figure \ref{fig:CVres} where the black curve denotes convergence result when the Lipschitz constant $L_\alpha$ has been divided by 100. We can observe that in this case, it suffices 1500 iterations to reach the precision obtained by the case $\|.\|_E = \|.\|_{\ell^1}$, $\|.\|_F=\|.\|_{\ell^2}$ and $d=\frac{1}{2}\|.\|_{\ell^2}^2$ (in red) after 29000 iterations.
Let us give an intuitive explanation to this positive behaviour. To simplify the reasoning, let us assume that $\pib^\star$ is the uniform probability distribution. First notice that the choice of $\|\cdot\|_E$ only influences the Lipschitz constant of $\nabla J_\alpha$ but does not change the algorithm, so that we can play with the norm on $E$ to decrease the local Lipschitz constant.  Furthermore, the choice of $\|\cdot\|_E$ minimizing the global Lipschitz constant may be different from the one minimizing the local Lipschitz constant. Considering that $\|\cdot\|_E = \|\cdot\|_{\ell^2}$, from Equation \eqref{eq:localLipschitz}, we get that $L_\alpha(\qb^\star) =   \frac{\|\Mb^* \|_{2 \rightarrow 2}^2}{\alpha \sigma_\Ec(\pib^\star)}$. Using Perron-Frobenius theorem, it can be shown that $\|\Mb^* \|_{2 \rightarrow 2}^2 = O(1)$ for our choice of dictionary, and $\sigma_\Ec(\pib^\star) = m$ for $ \|\cdot\|_E = \|\cdot\|_{\ell^2}$. From this simple reasoning, we can conclude that the local Lipschitz constant around $\pib^\star$ is no greater than $O(1/m)$.
This means that if the minimizer is sufficiently far away from the simplex boundary, we can decrease $L_\alpha$ by a significant factor without loosing convergence.
\begin{figure}
\begin{center}
\includegraphics[height=9cm]{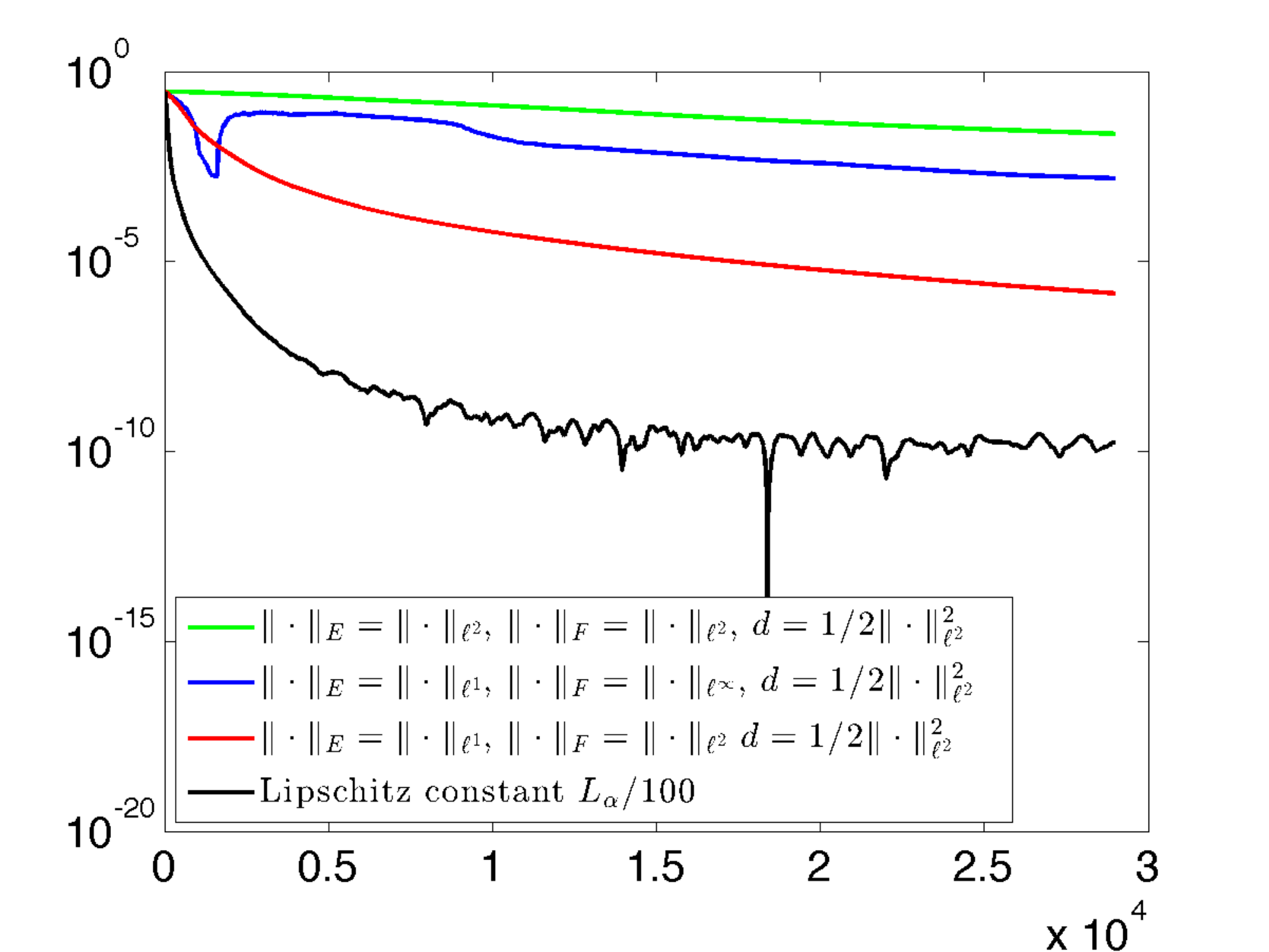}
\caption{\label{fig:CVres} Convergence curves in a semi-logarithmic scale for Algorithm \ref{alg:nesterov} ($\alpha=10^{-2}$) (number of iterations on the $x$-axis) in green the case where $\|.\|_E = \|.\|_{\ell^2}$, $\|.\|_F=\|.\|_{\ell^2}$, $d=\frac{1}{2}\|.\|_{\ell^2}^2$, in red the case where $\|.\|_E = \|.\|_{\ell^1}$, $\|.\|_F=\|.\|_{\ell^2}$, $d=\frac{1}{2}\|.\|_{\ell^2}^2$, in blue the case where $\|.\|_E = \|.\|_{\ell^1}$, $\|.\|_F=\|.\|_{\ell^\infty}$ $d=\frac{1}{2}\|.\|_{\ell^2}^2$, and in black the case where $\|.\|_E = \|.\|_{\ell^1}$, $\|.\|_F=\|.\|_{\ell^2}$, $d=\frac{1}{2}\|.\|_{\ell^2}^2$ with a restricted Lipschitz constant $L_\alpha'=L_\alpha/100$.  }
\end{center}
\end{figure}

\section{Numerical results}
\label{sec:num}
In this section, we assess the reconstruction performance of the sampling patterns using the approach described in Section \ref{subsec:resDual} with $\alpha=10^{-2}$. We compare it to standard approaches used in the context of MRI. We call $\pib\left[\pb\right]$ the probability distribution $\pib^\star$ resulting from the minimization problem \eqref{pb:finalGoal} for a given target distribution $\pb$ on isolated measurements.

\subsection{The choice of a particular dictionary of blocks}

\label{subsec:dico}

From a numerical point of view, we study a particular system of blocks of measurements. 
The dictionary used in all numerical experiments of this article is composed of discrete lines of length $\ell$, joining any pixel on the edge of the image to any pixel on the opposite edge, as in Figure \ref{fig:blocks}(b). Note that the number of blocks in this dictionary is $n_1^2 + n_2^2$ for an image of size $n_1\times n_2$. The choice of such a dictionary is particularly relevant in MRI, since the gradient waveforms that define the acquisition paths is subject to bounded-gradient and slew-rate constraints, see e.g. \cite{lustig2008fast}. Moreover the practical implementation on the scanner of straight lines is straightforward since it is already in use in standard echo-planar imaging strategies. 

Remark that, in such a setting, the mapping $\Mb$, defined in \eqref{eq:M}, is a linear mapping that can be represented by a matrix of size $n \times \nbloc$ with $\Mb_{i,j} = 1/\ell$ when the $i$-th pixel belongs to the $j$-th block, for $i=1,\hdots , n $ and $j = 1, \hdots , \nbloc$.

One may argue that in MRI, dealing with samples lying on continuous lines (and not discrete grids) is more realistic in the design of the MR sequences. To deal with this issue, one could resort to the use of the Non-Uniform Fast Fourier Transform. This technique is however much more computationally intensive. In this paper we thus stick to values of the Fourier transform located on the Euclidean grid.  This is commonly used in MRI with regridding techniques.

\subsection{The reconstructed probability distribution}

We are willing to illustrate the fidelity of $\pib\left[\pb\right]$, the solution of Problem \eqref{pb:finalGoal}, to a given target $\pb$.
In the setting of 2D MR sensing, with the dictionary of lines in dimension $n_1 = n_2 = 256$ described in the previous subsection. We set the target probability  distribution $\pb=\pb_\text{opt}$ the one suggested by current CS theories on the set of isolated measurements. It is proportional to $\|\ab_k^*\|_{\ell^\infty}^2$, see \cite{puy2011variable,chauffert2013variable,bigot2013analysis}. To give an idea of what the resulting probability distribution $\pib\left[\pb_{\text{opt}}\right]$ looks like, we draw 100~000 independent blocks of measurements according to $\pib\left[\pb_{\text{opt}}\right]$ and count the number of measurement for each discrete Fourier coefficient. The result is displayed on Figure \ref{fig:compProbaTarget}. This experiment underlines that our strategy manages to catch the overall structure of the target probability distibution. 

\begin{figure}[h]
\begin{center}
\btabu{@{}cc}
\includegraphics[scale=0.4]{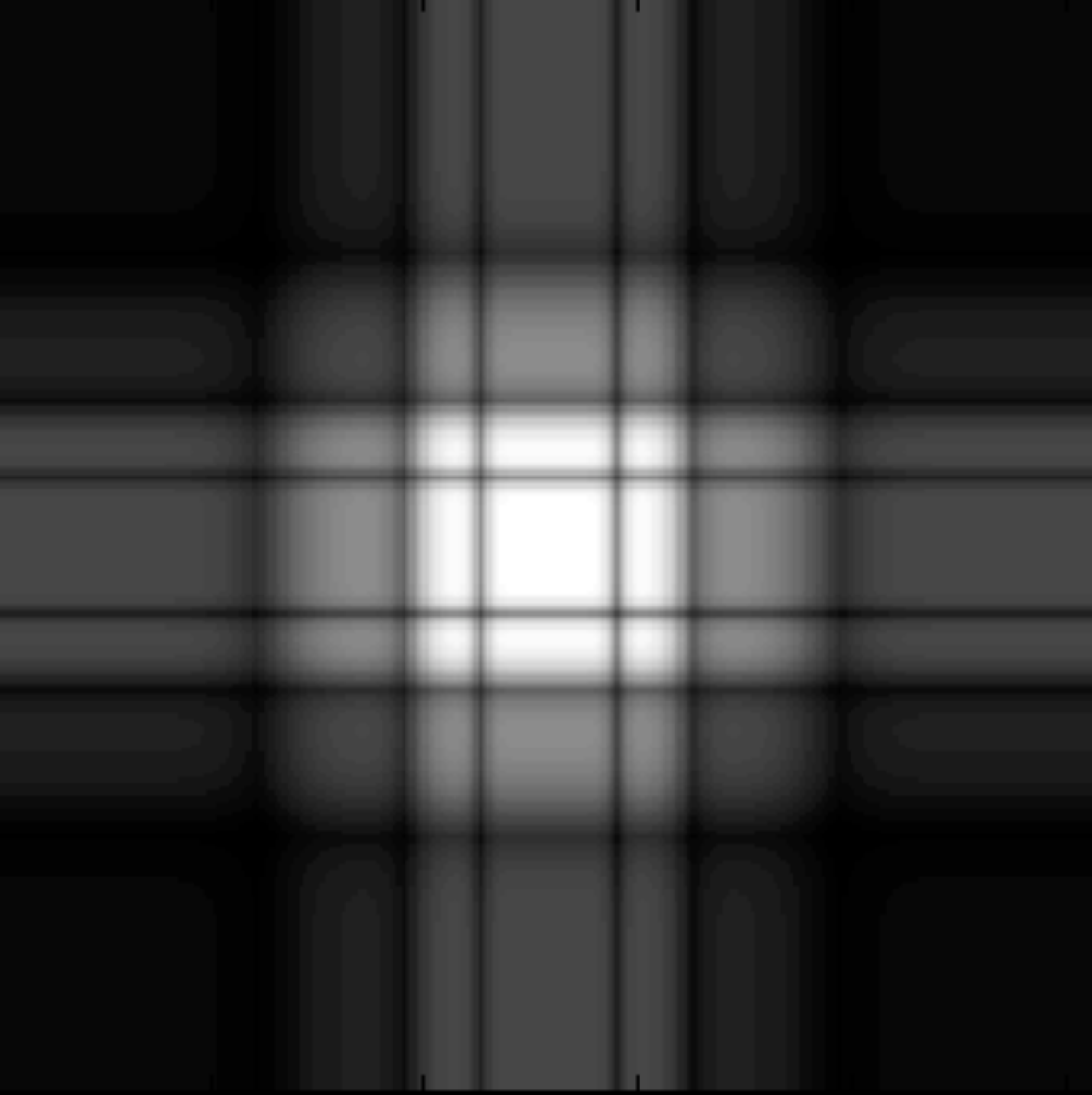} &
\includegraphics[scale=0.4]{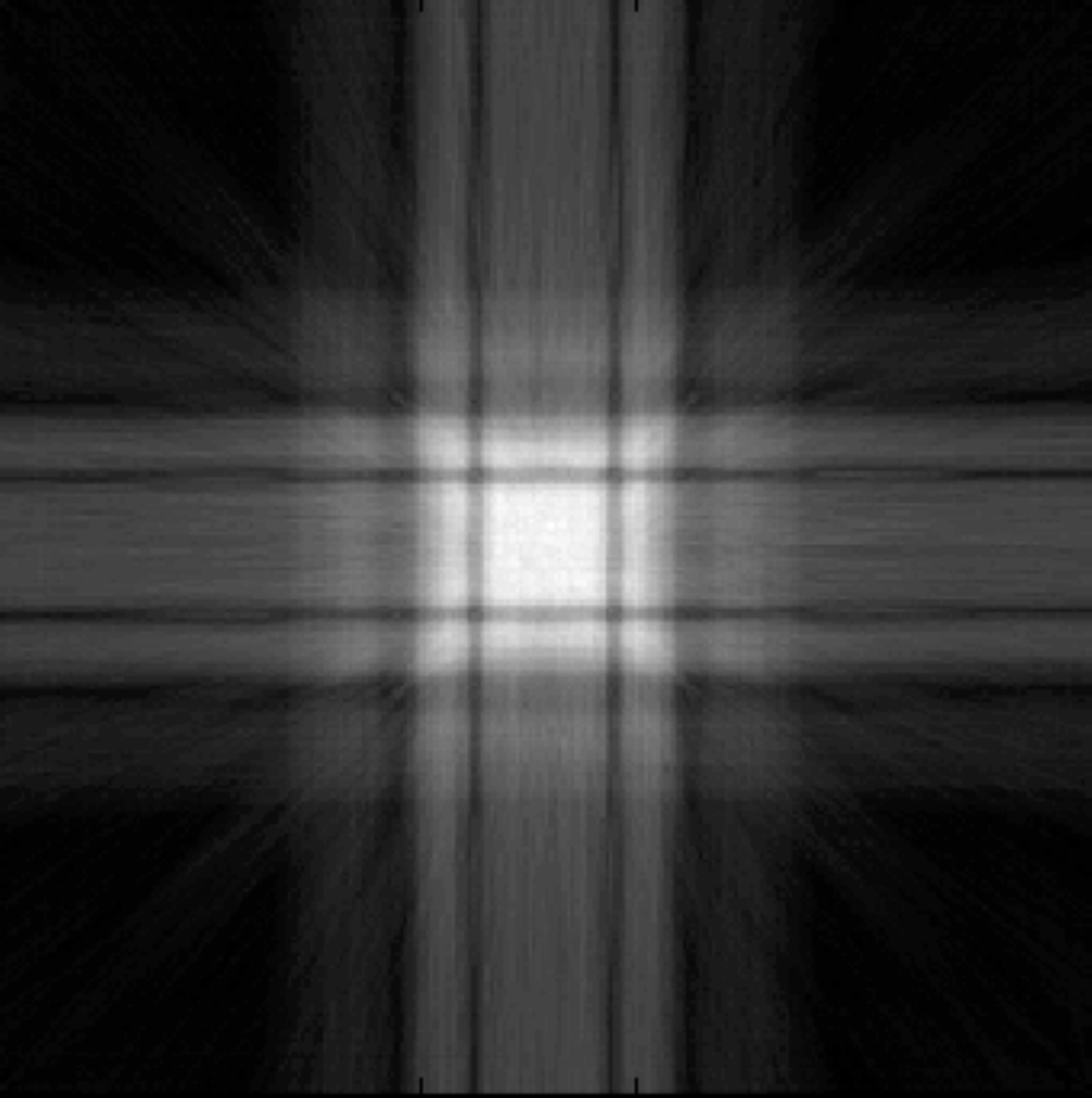} \\
{\small (a)}&{\small (b)} 
\etabu
\caption{\label{fig:compProbaTarget} Illustration of the fidelity of $\pib\left[\pb_{\text{opt}}\right]$ to $\pb_{\text{opt}}$. (a): on the left hand side, we present the target probability distribution $\pb_{\text{opt}}$ (b): on the right hand side, we perform 100000 i.i.d.\ drawings according to $\pib\left[\pb_{\text{opt}}\right]$ of blocks from the blocks dictionary and count the number of times that a point is sampled at each location.}
\end{center}
\end{figure} 

\subsection{Reconstruction results}

In this section, we compare the reconstruction quality of MR images for different acquisition schemes. The comparison is always performed for schemes with an equivalent number of isolated measurements.
We recall that in the case of MR images, the acquisition is done in the Fourier domain, and MR images are supposed to be sparse in the wavelet domain. Therefore, the full sensing matrix $\Ab = \left( \ab_1 | \ab_2 | \hdots | \ab_n\right)^*$, which models the acquisition process, is the composition of a Fourier transform with an inverse Wavelet transform. The reconstruction is done via $\ell^1$-minimization as presented in \eqref{pbMin1}, using Douglas-Rachford algorithm \cite{combettes2011proximal}.
It was proven in various papers \cite{chauffert2013variable,chauffertSIAM,adcock2013breaking} that MRI image quality can be strongly improved by fully acquiring the center of the Fourier domain via a mask defined by the support of the mother wavelet, see Figure \ref{fig:mask}. Therefore, for every type of schemes used in our reconstruction test, we first fully acquire this mask. 

\begin{figure}[h]
\begin{center}
\btabu{@{}cc}
\includegraphics[height=4cm]{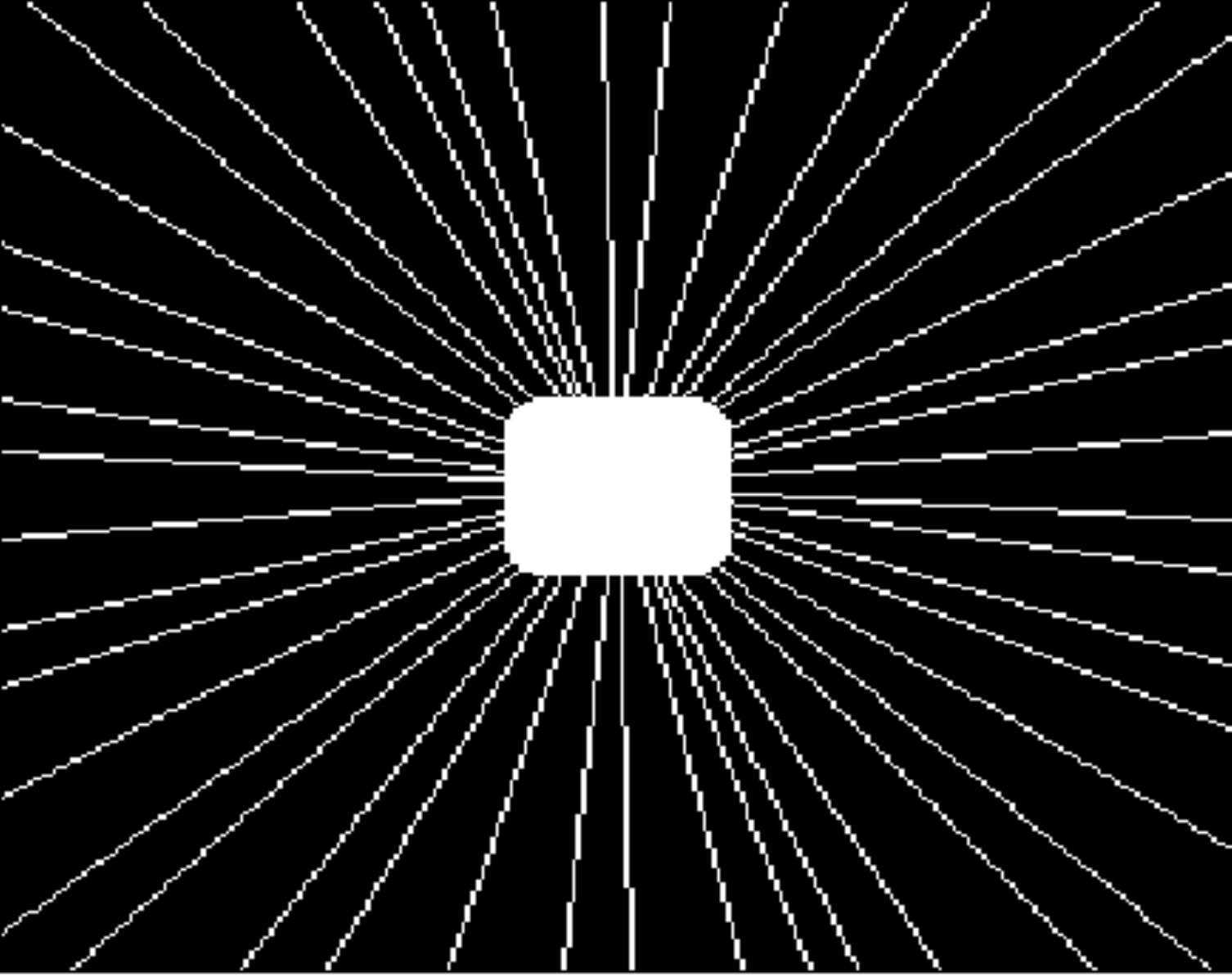} &
\includegraphics[height=4cm]{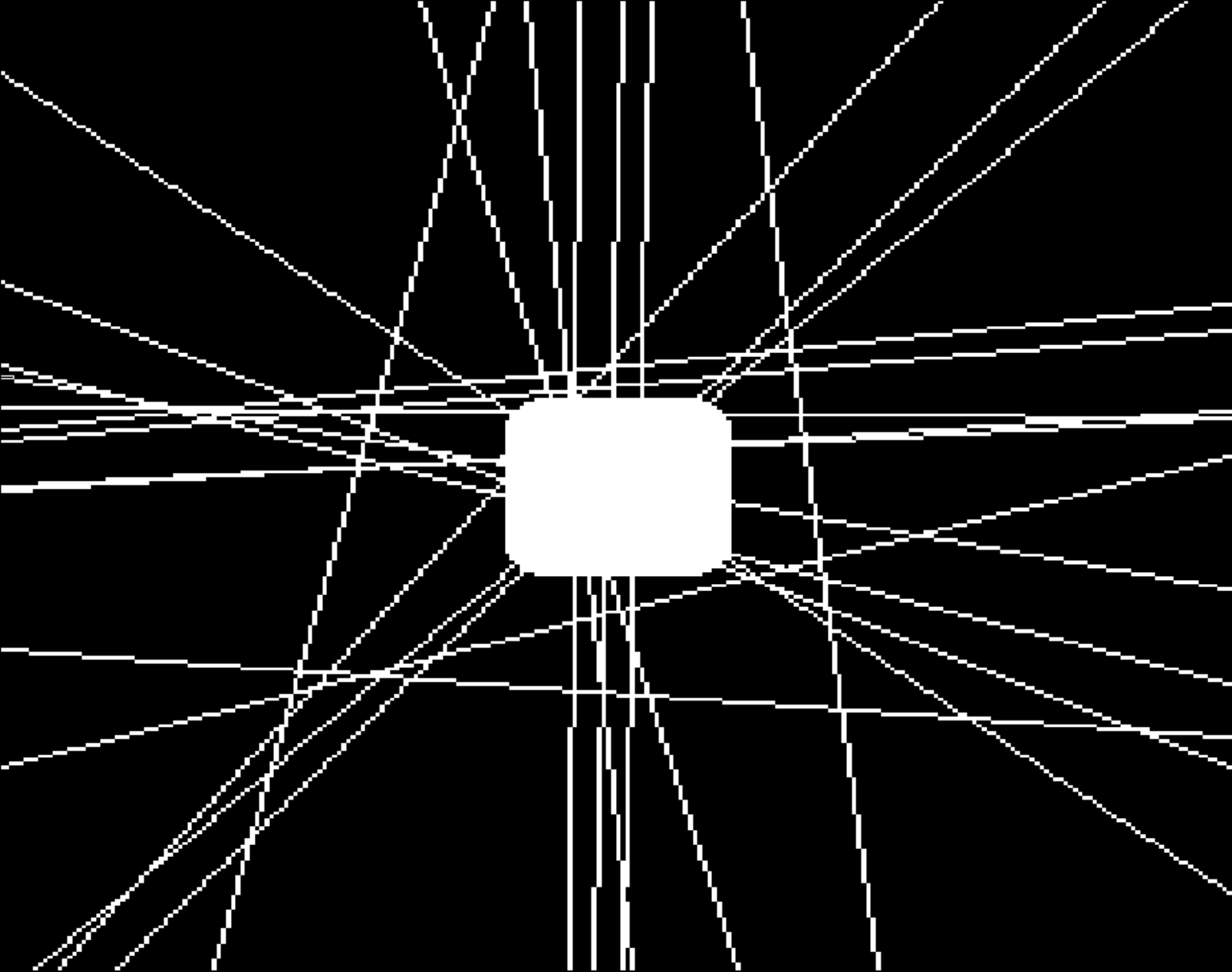} \\
{\small (a)}&{\small (b)} 
\etabu
\caption{\label{fig:mask}Different schemes based  (a) on the golden angle pattern, and (b) on the dictionary proposed in Section \ref{subsec:dico}. Both schemes are combined with a mask which fully samples the center of the Fourier domain. In both cases, the proportion of total measurements represents $10\%$ of the full image, while the mask defined by the support of the mother wavelet represents $3\%$ of the full image. }
\end{center}
\end{figure}
\FloatBarrier

The various schemes considered in this paper are based on blocks of measurements and on heuristic schemes that are widely used  in the context of MRI. They will consist in:
	\begin{itemize}
	\item Equiangularly distributed radial lines: the scheme is made of lines always intersecting the center of the acquisition domain, and that are distributed uniformly \cite{lustig2008compressed}.
	\item Golden angle scheme: the scheme is made of radial lines separated from the golden angle, i.e $111.246^\circ$. This technique is used often in MRI sequences, and it gives good reconstruction results in practice \cite{winkelmann2007optimal}.
	\item Random radial scheme: radial lines are drawn uniformly at random \cite{chan2012influence}.
	\item Scheme based on the dictionary described in Section \ref{subsec:dico} 
			\begin{itemize}
			\item Blocks are drawn according to $\pib\left[ \pb_{\text{rad}}\right]$ which is the resulting probability distribution obtained by minimizing Problem \eqref{pb:finalGoal} for $\pb = \pb_{\text{rad}}$. The distribution $\pb_{\text{rad}}$ a radial distribution that decreases as $\mathcal{O}\left( \frac{1}{k_x^2+k_y^2}\right)$. This choice was justified recently in \cite{krahmer2012beyond} and used extensively in practice. Note that $\pb_{\text{rad}}$ is set to $0$ on the $k$-space center since it is already sampled deterministically, see Figure \ref{fig:ProbaTarget} (b).
			\item Blocks are drawn according to $\pib\left[\pb_{\text{opt}}\right]$, where $\pb_{\text{opt}}$ is  defined by \eqref{eq:popt}, which is the resulting probability distribution obtained by minimizing Problem \eqref{pb:finalGoal} for $\pb = \pb_{\text{opt}}$ defined in \cite{chauffert2013variable,bigot2013analysis}. Once again, $\pb_{\text{opt}}$ is set to $0$ on the $k$-space center, see Figure \ref{fig:ProbaTarget} (a).
			\end{itemize} 
	\end{itemize}

\begin{figure}
\begin{center}
\btabu{@{}cc}
\includegraphics[height=6cm]{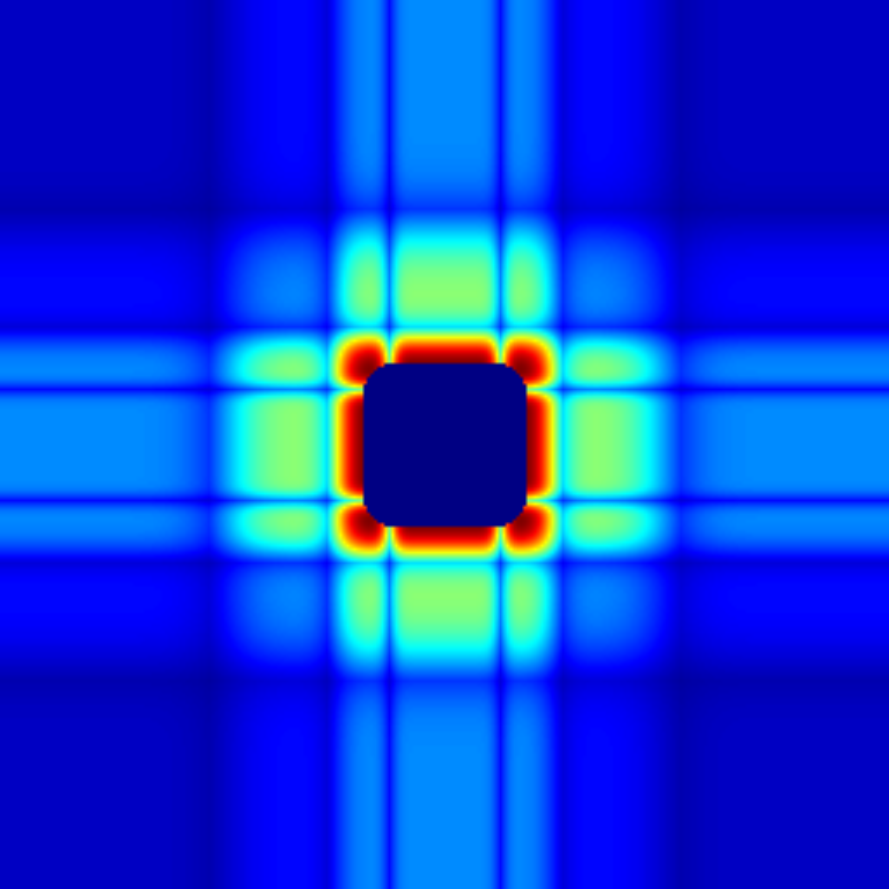} &
\includegraphics[height=6cm]{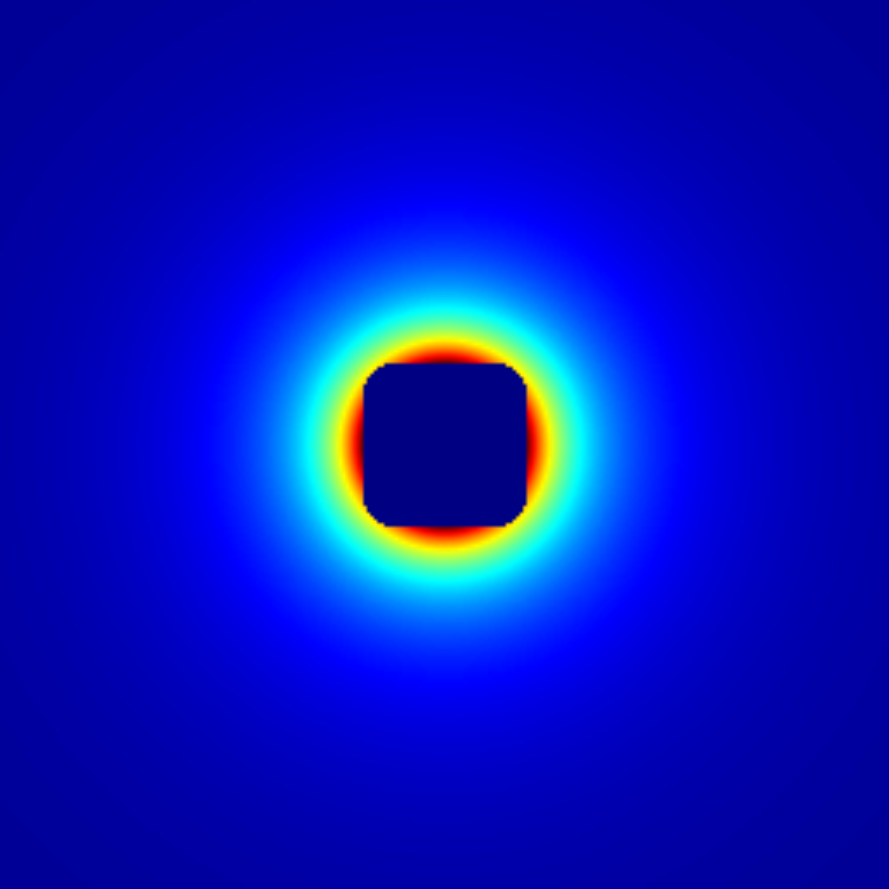} \\
{\small (a)}&{\small (b)} 
\etabu
\caption{\label{fig:ProbaTarget}Target probabilities on pixels (in red, high values, and in dark blue, values close to 0). (a) displays the distribution proportional to $\|\ab_i^*\|_{\ell^\infty}^2$ defined in \cite{chauffert2013variable}, (b) displays a radial distribution as presented in \cite{krahmer2012beyond}. The center has been set to zero, since it will be sampled by the mask in a deterministic way.}
\end{center}
\end{figure}

\begin{figure}
\begin{center}
\btabu{@{}cc}
\includegraphics[height=6cm]{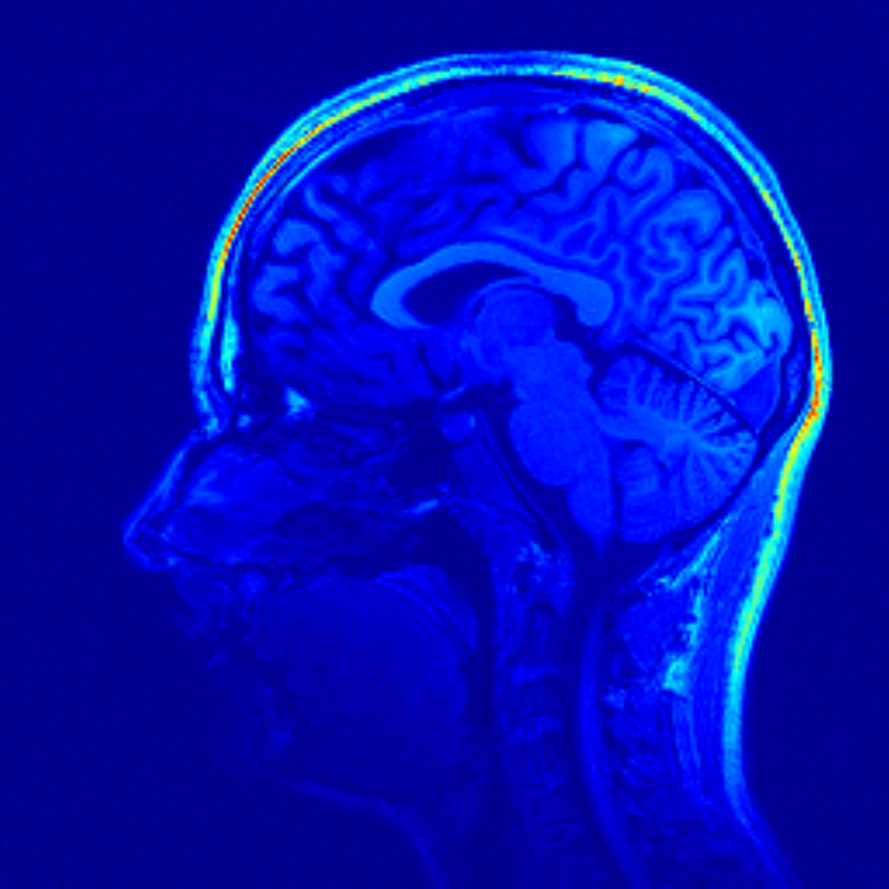} &
\includegraphics[height=6cm]{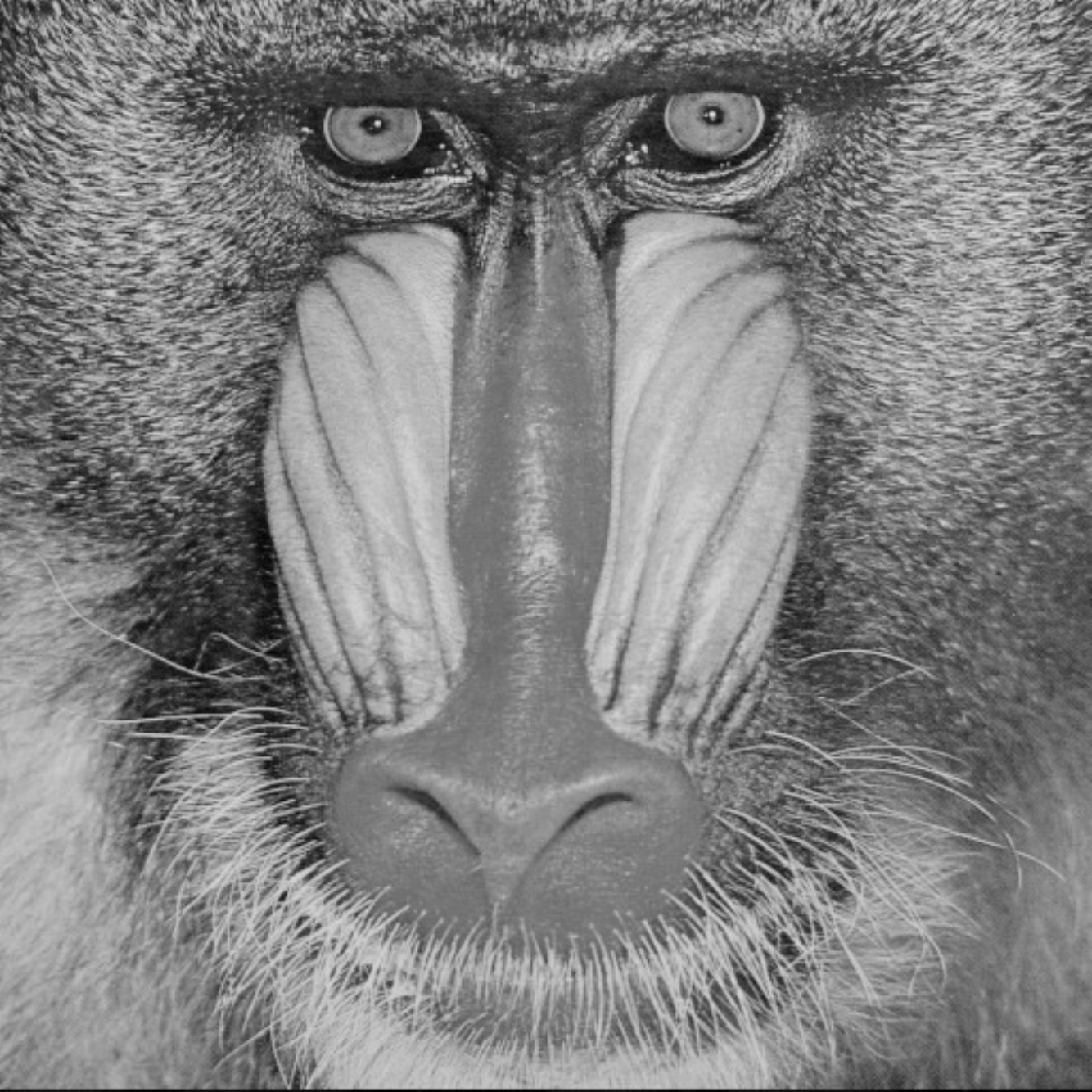} \\ 
{\small (a) Brain}&{\small (b) Baboon} 
\etabu
\caption{\label{fig:references} Reference images to reconstruct for the settings $256\times 256$ and $512\times 512$ .} 
\end{center}
\end{figure}

\subsubsection*{Setting $256 \times 256$} 

The numerical experiment is run for images of size $n_0 \times n_0$ with $n_0=256$. The full dictionary described in Section \ref{subsec:dico} contains lines of length $\ell = n_0$ pixels connecting every point on the edge of the image to every point on the opposite side. 
For each proportion of measurements ($10\%, 15\%, 20\%, 25\%, 30\%, 40\%, 50 \%$), we proceed to $100$ drawings of schemes when the considered scheme is random. Reconstruction results, for the reference images showed in Figure \ref{fig:references} and for various sampling schemes, are displayed in the form of boxplots of PSNR in Figure \ref{fig:resultatsReconstruction} (a)(c).

\begin{figure}
\begin{center}
\btabu{@{}cc}
\hspace{-2cm}
\includegraphics[height=8cm]{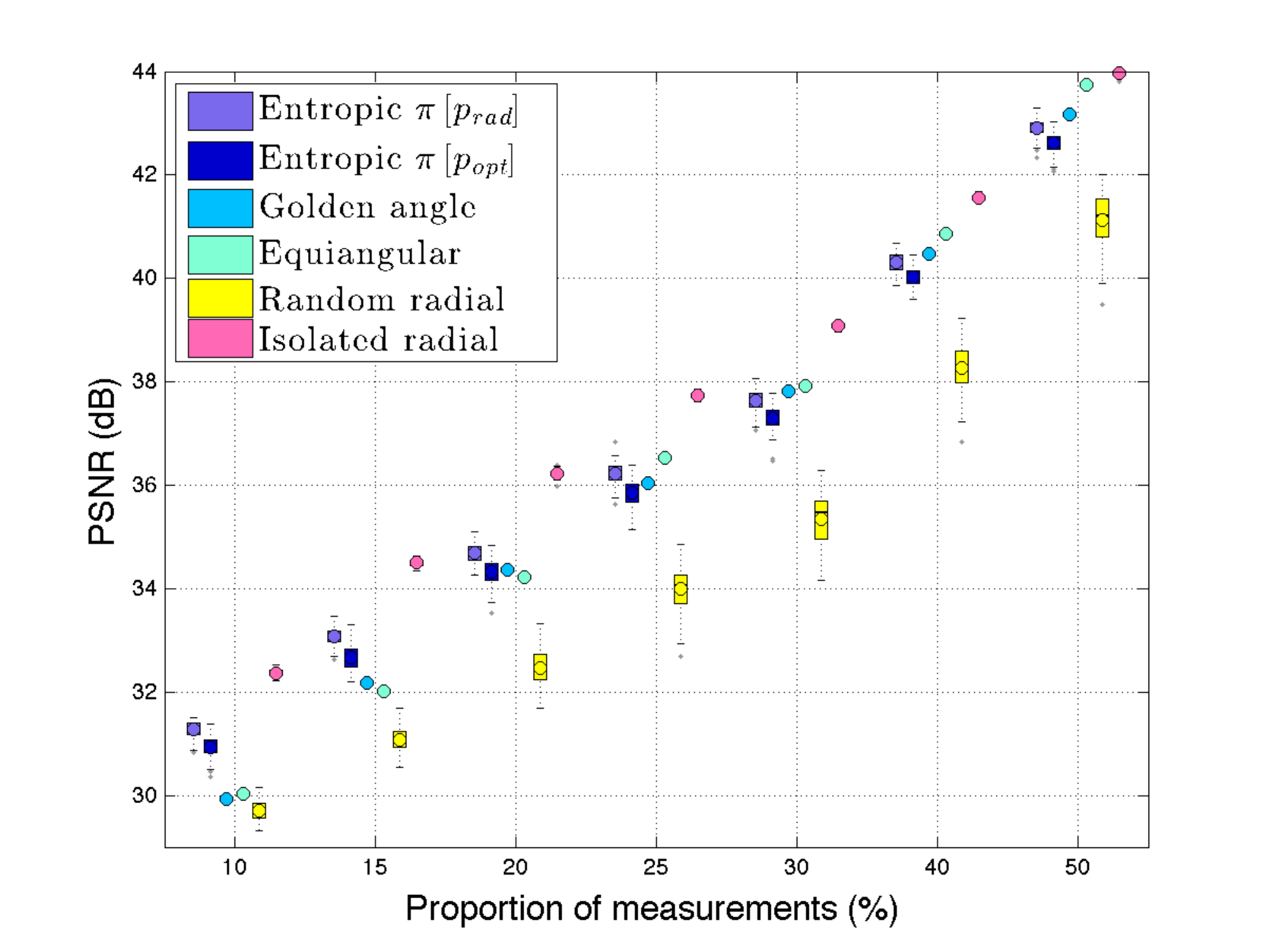} \hspace{-1.5cm}&
\hspace{-1.5cm}
\includegraphics[height=8cm]{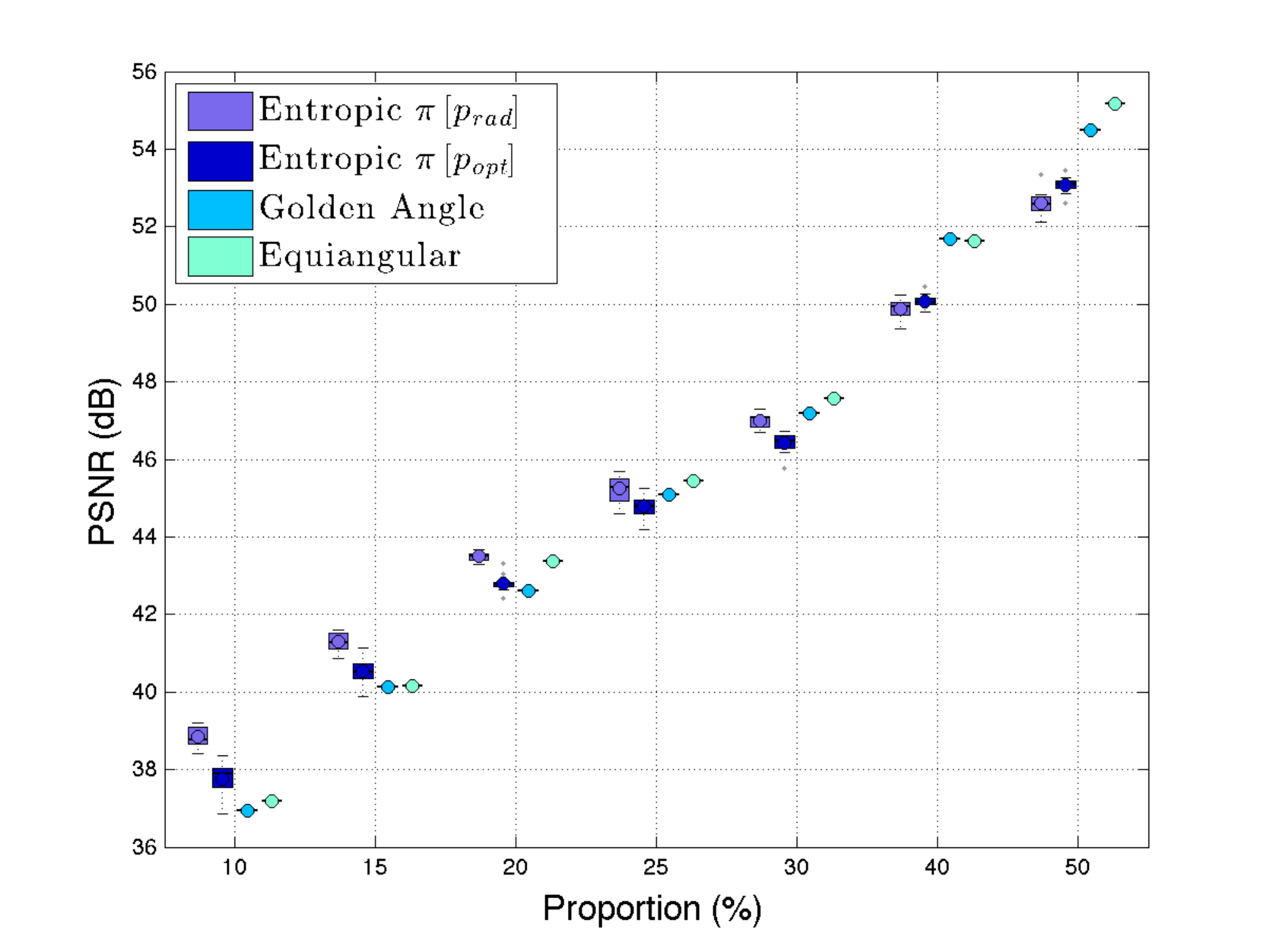} \\
{\hspace{-2cm}\small (a) Brain in $256\times 256$ \hspace{-1.5cm}}&{\hspace{-1.5cm}\small (b) Brain in $512\times 512$}  \\
\hspace{-2cm}
\includegraphics[height=8cm]{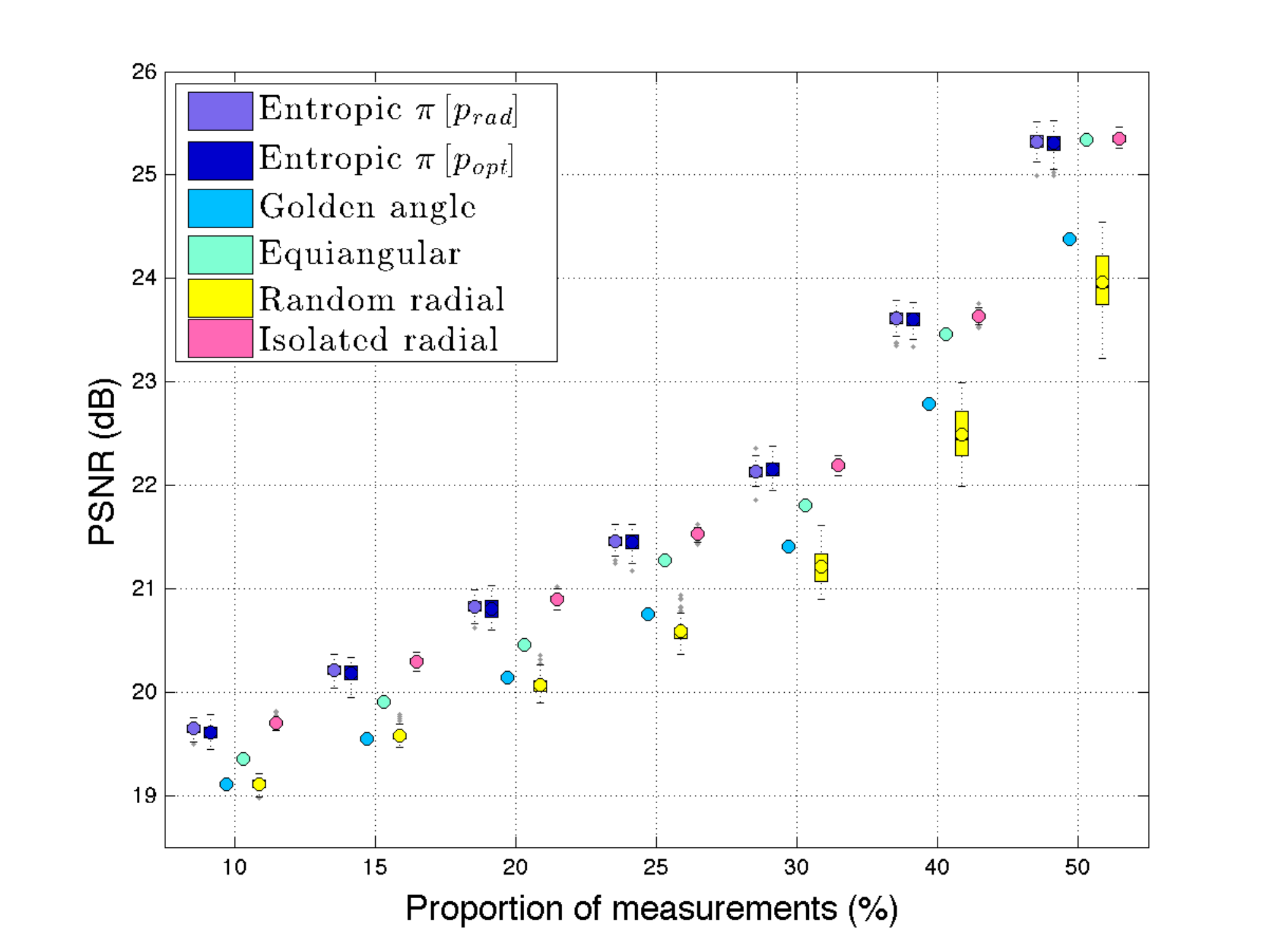} \hspace{-1.5cm}&
\hspace{-1.5cm}\includegraphics[height=8cm]{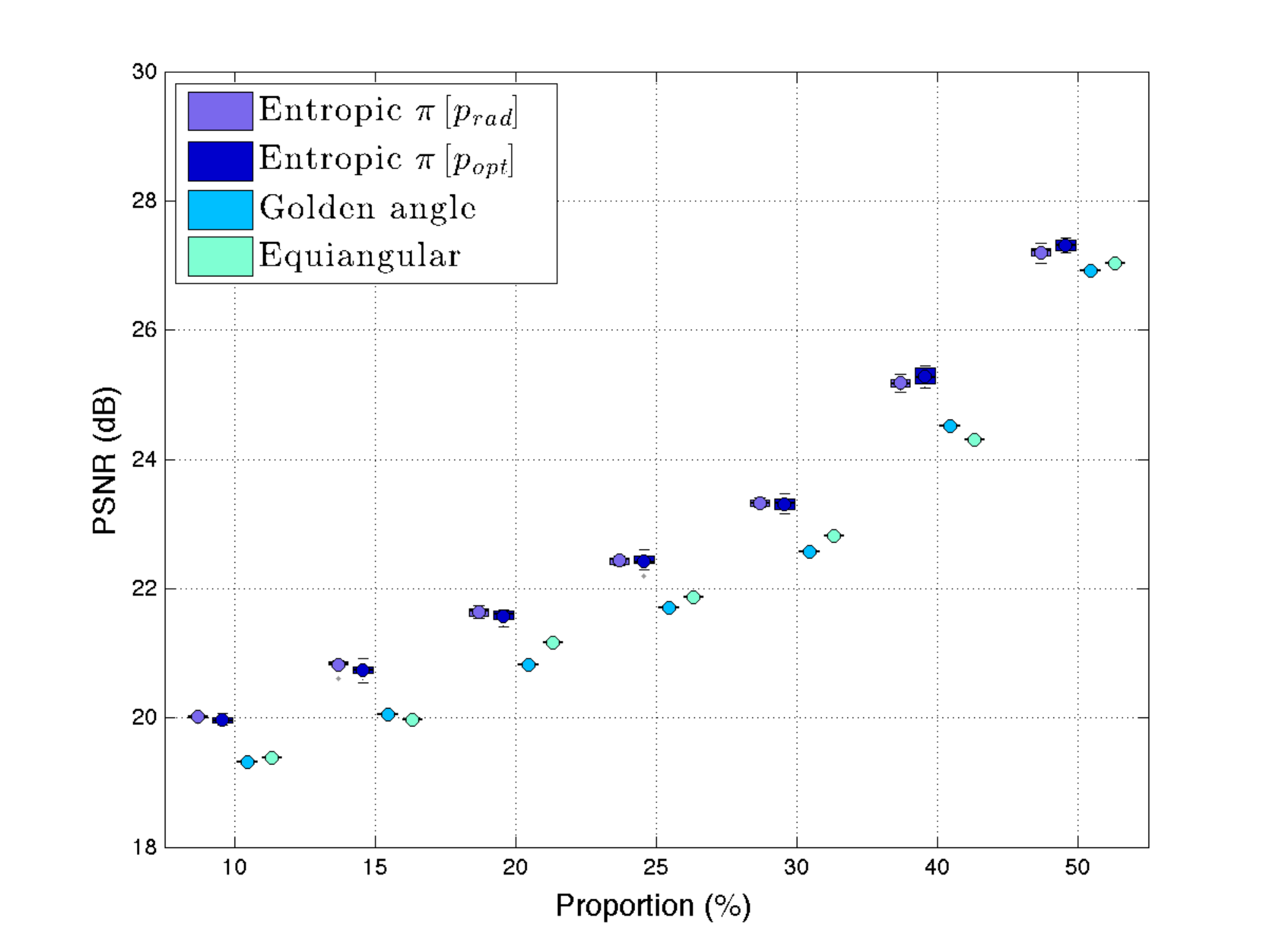} \\
{\hspace{-2cm}\small (c) Baboon in $256\times 256$\hspace{-1.5cm}}&{\hspace{-1.5cm}\small (d) Baboon in $512\times 512$}  \\
\etabu
\caption{\label{fig:resultatsReconstruction}\small{Box plots for PSNR of the reconstructed images (brain, baboon) with respect to the proportion of measurements chosen in the scheme ($10\%, 15\%, 20\%, 25\%, 30\%, 40\%, 50 \%$). The undersampling ratio for all schemes is the ratio between the number of sampled distinct frequencies and the total number of possible measurements. This means that duplicated frequencies are accounted for only once.}}
\end{center}
\end{figure}

Figure \ref{fig:resultatsReconstruction} shows that the schemes based on the approach presented in this article give better results than random radial schemes, for any proportion of measurements. The improvement in terms of PSNR is generally between $1$ and $2$ dB. The schemes based on $\pib\left[\pb_{\text{opt}}\right]$ and  $\pib\left[\pb_{\text{rad}}\right]$ are competitive with those based on the golden angle or equiangularly distributed schemes in the case where the proportion of measurements is low (less than $20\%$ of measurements). We observe that for $10\%$ measurements, schemes based on our dictionary and drawn according to $\pib\left[\pb_{\text{rad}}\right]$ outperform by more than $1$ dB the standard sampling strategies. Increasing the PSNR of $1$dB is significant and can be qualitatively observed in the reconstructed image.

Figures \ref{fig:resultatsReconstruction}(a) and (c) allow to compare the quality of the reconstructions using different sampling schemes and different undersampling ratios. 
In this experiment, it can be seen that block-constrained acquisition never outperforms acquisitions based on indepenedent measurements. This was to be expected since adding constraints reduces the space of possible sampling patterns. Once again, note that independent drawings are however not conceivable in many contexts such as MRI. In Figures \ref{fig:resultatsReconstruction}(a) and (c), it can also be seen that the proposed sampling approach always produces results comparable to the standard sampling schemes and tend to produce better results for low sampling ratios.

Finally, in Figure \ref{fig:compTirageProba}, we illustrate that block-constrained acquisition does not allow to reach an arbitrary target distribution by showing the difference between $\pb_{\text{rad}}$ and the probability distribution $\Mb\left(\pib\left[\pb_{\text{rad}}\right]\right)$ which is defined on the set of isolated measurements. This confirms Proposition \ref{prop:subset} experimentally. 

\begin{figure}
\begin{center}
\btabu{@{}cc}
\includegraphics[height=7cm]{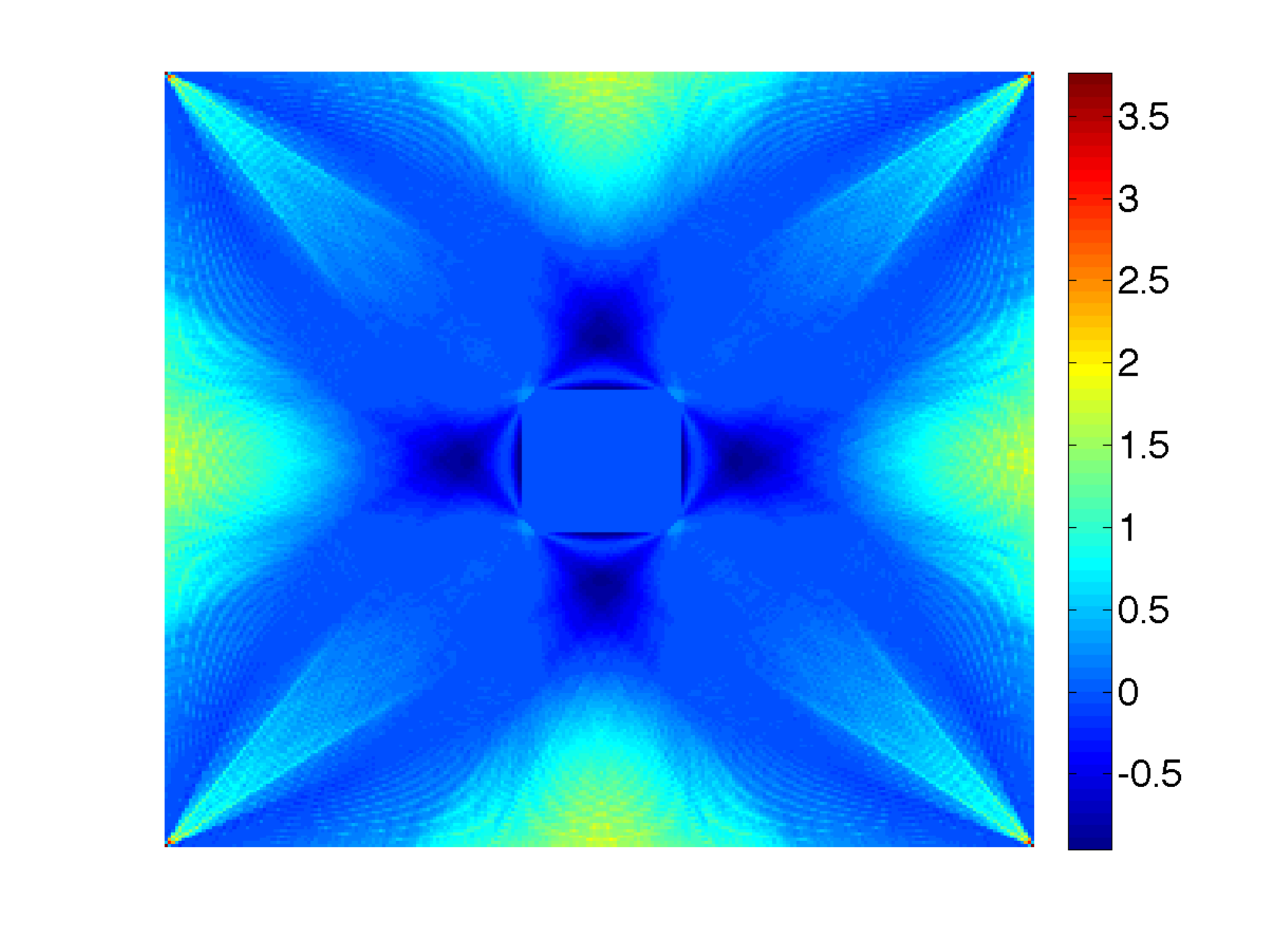} & \hspace{-1cm}
\includegraphics[height=7cm]{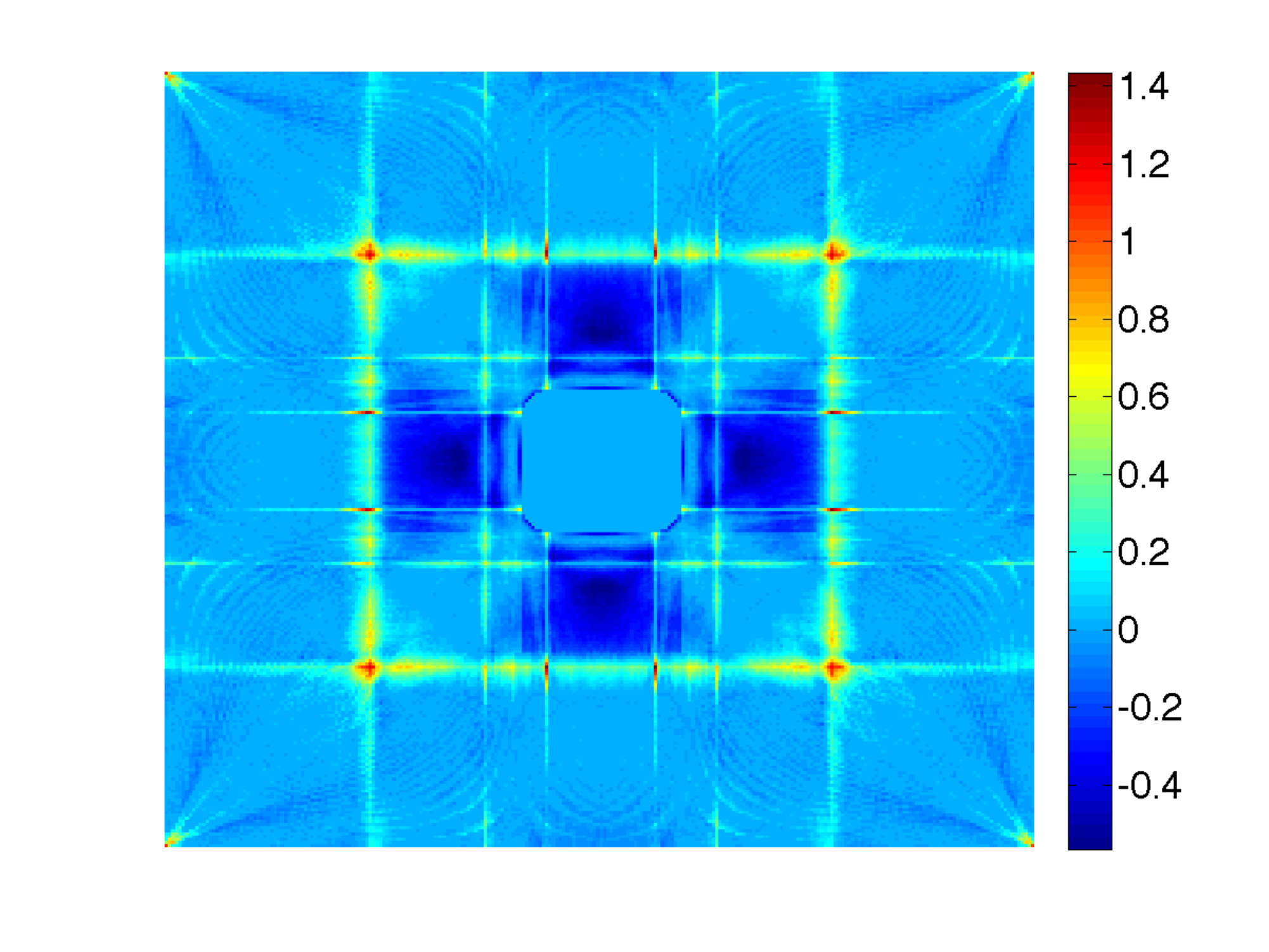} \\
{\small (a)}&\hspace{-1cm}{\small (b)} 
\etabu
\caption{\label{fig:compTirageProba} Difference between the target probabilities $\pb$ and $\Mb \pib(\pb)$ relatively to the magnitude of $\pb$, i.e.\ we show the following quantity $\frac{ \left(\Mb(\pib(\pb))\right)_i - \pb_i}{\pb_i}$ for the $i$-th sampling location, (a) for the radial distribution $\pb_{\text{rad}}$, we see that we "sub-draw" by a factor of 50 \% around the mask, and we "over-draw"  by a factor of 150 \% at the center of the edges. (b) for the CS optimal distribution $\pb_{\text{opt}}$, we see that we "sub-draw" by a factor of 40 \% around the mask. Note that the sub-drawing effect cannot be avoided: indeed, we cannot reach any target probability distribution via $\Mb$, see Proposition \ref{prop:subset}.}
\end{center}
\end{figure}

\subsubsection*{Setting $512\times 512$} 

Given that in CS the quality of the reconstruction can be resolution dependent, as described in \cite{adcock2013breaking}, we have decided to run the same numerical experiment on $512 \times 512$ images. The numerical experiment is run for images of size $n_0 \times n_0$ with $n_0=512$. The full dictionary described in Section \ref{subsec:dico} contains lines of length $\ell = n_0$ connecting every point on the edge of the image to every point on the opposite side. For each proportion of measurements ($10\%, 15\%, 20\%, 25\%, 30\%, 40\%, 50 \%$), we proceed to $10$ drawings of sampling schemes when the considered scheme is random.
The images of reference to reconstruct are the same as in the setting $256\times 256$, see Figure \ref{fig:references}.

\begin{figure}
\begin{center}
\btabu{@{}cc}
\includegraphics[height=7cm]{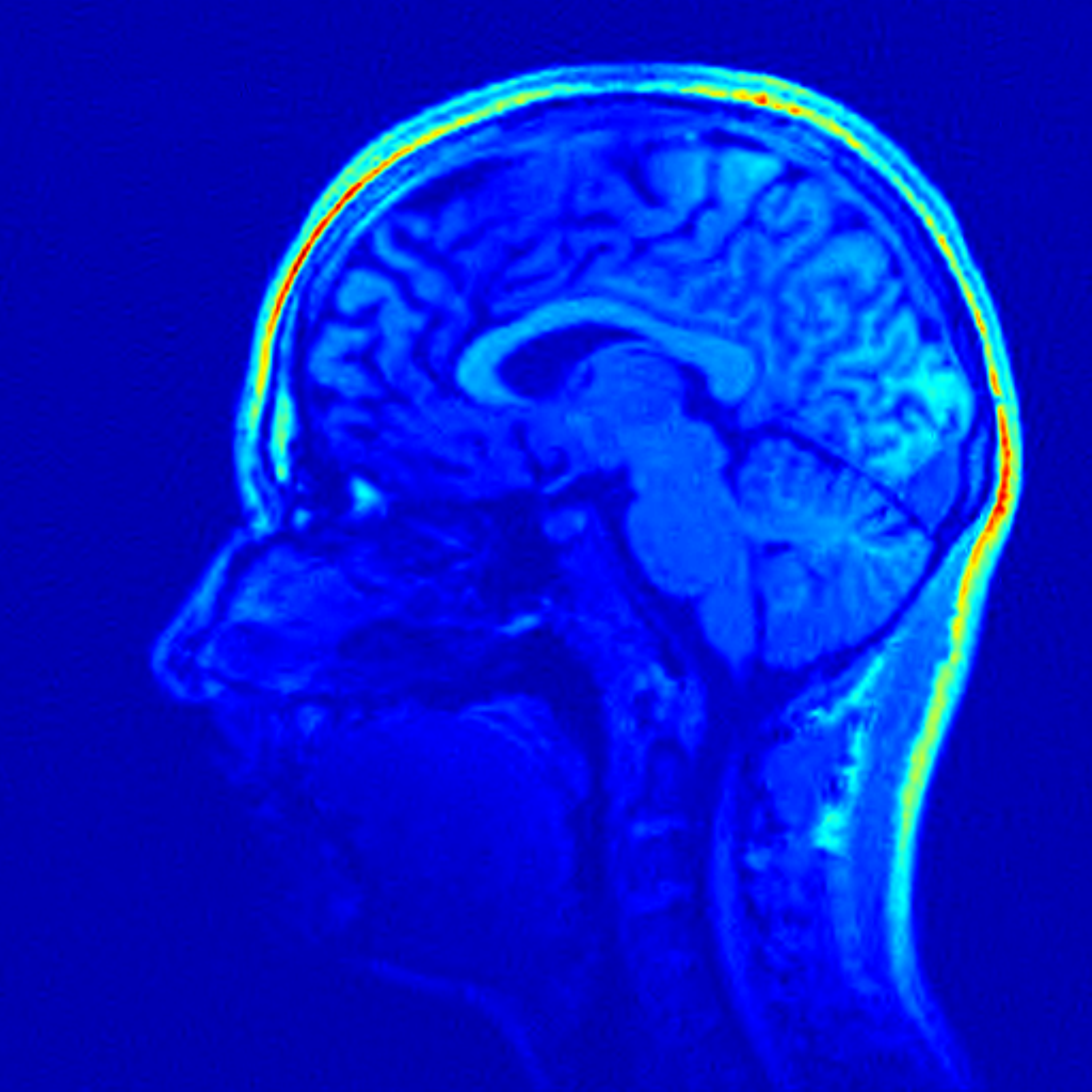} & 
\includegraphics[height=7cm]{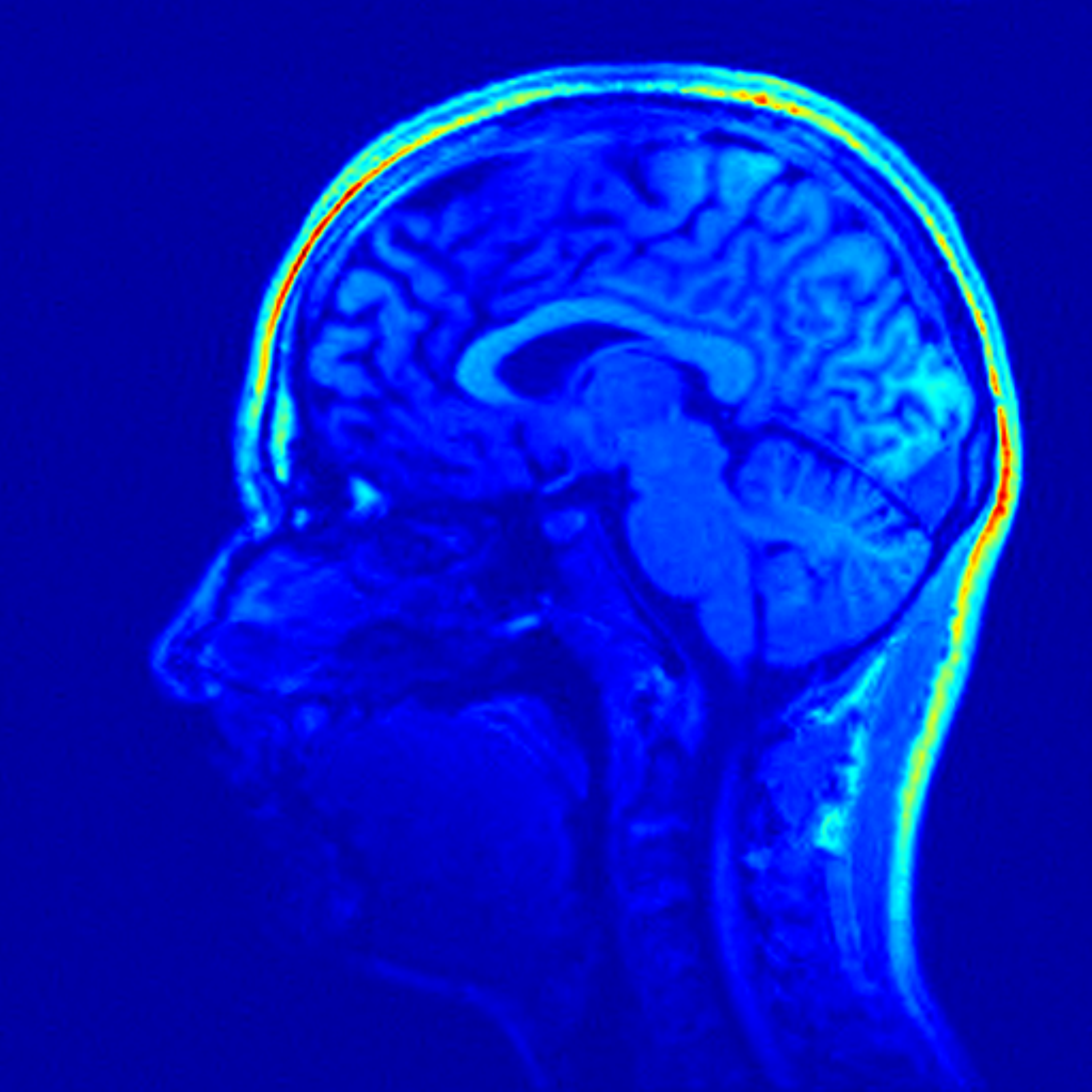} \\
{\small (a) PSNR = 40.1364 dB}&{\small (b) PSNR = 41.8854 dB} \\
\includegraphics[height=7cm]{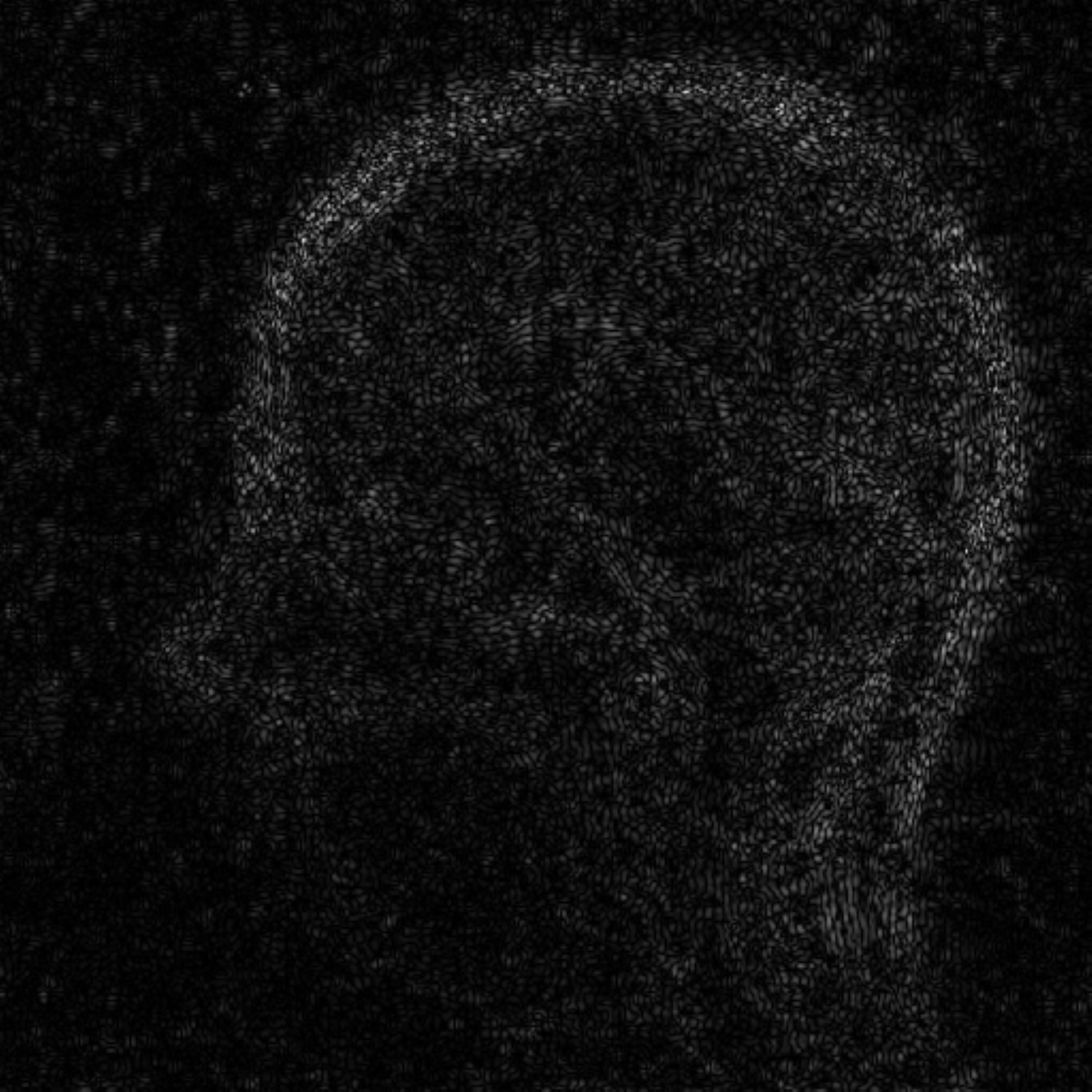} & 
\includegraphics[height=7cm]{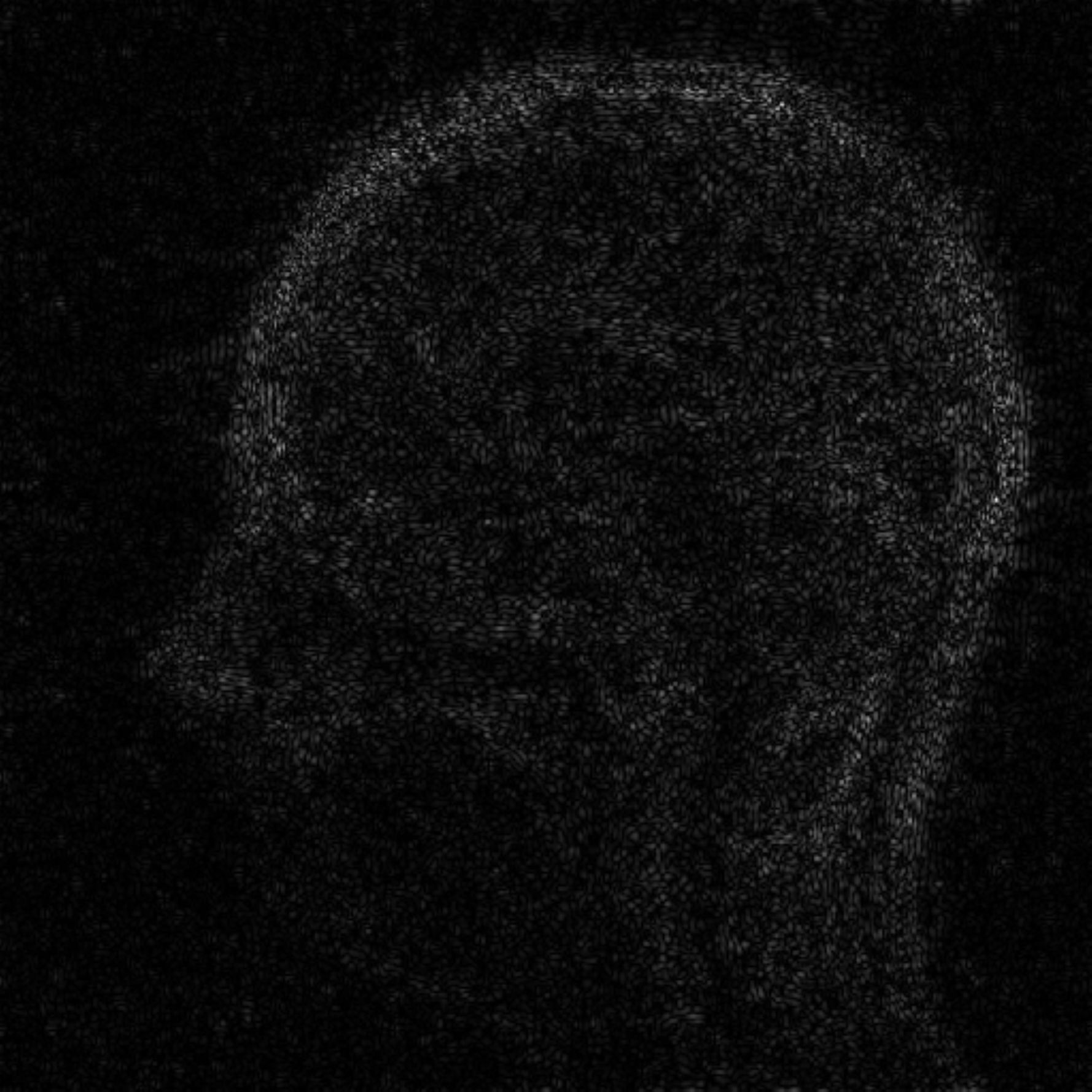} \\
{\small (c) }&{\small (d) } \\ 
\etabu
\caption{\label{fig:reconstruction512_15} Comparison of the reconstructed images from $15\%$ of measurements for a $512\times512$ image for a golden angle scheme (a), and a scheme based on our dictionary and $\pib(\pb_{\text{rad}})$ (b). We respectively plot the absolute difference to the reference image for the reconstruction using a golden angle scheme in (c) and  for the reconstruction using a scheme based on $\pib(\pb_{\text{rad}})$ in (d). Note that in (c) and (d), the gray levels are in the same scale.}
\end{center}
\end{figure}

Quality of reconstructions are compared in Figure \ref{fig:resultatsReconstruction}(b) and (d) for the golden or equiangularly distributed lines and our proposed method based on $\pib( \pb_{\text{opt}})$ and $\pib( \pb_{\text{rad}})$. This experiment shows that the PSNR of the reconstructed images is significantly improved by using the proposed method until $30\%$ of measurements for the brain image and until $40\%$ of measurements for the baboon one. We can remark that for both images, for a same proportion of measurements, the PSNR of the reconstructed images increases while the resolution increases. This numerical experiment suggests that the proposed sampling approach might be significantly better than traditional ones in the MRI context for high resolution images.
In Figure \ref{fig:reconstruction512_15}~(a), we present the reconstructed image of the brain from $15\%$ of measurements in the case of a golden angle scheme. In Figure \ref{fig:reconstruction512_15}~(b), we present the reconstructed image of the brain from $15\%$ of measurements in the case of a realization of schemes based on $\pib(\pb_\text{rad})$. The latter's PSNR is 41.88 dB whereas in the golden scheme case, the PSNR only reaches 40.13 dB. In Figure \ref{fig:reconstruction512_15}~(c) and (d), we display the corresponding difference images to the reference image, which underlines the improvement of 1.7 dB in our method.

As a side remark, let us mention that in MRI, sampling diagonal or horizontal lines actually takes the same scanning time (even though the diagonals are longer), since gradient coils work independently in each direction. In the MRI example, the length of the path is thus less meaningful that the number of scanned lines. In Figure \ref{fig:nbLines}, we show different sampling schemes based on the golden angle pattern or on our method with the same number of lines, and we show the corresponding reconstructions of brain images.

\begin{figure}
\begin{center}
\begin{tabular}{@{}cc}
\includegraphics[height=6cm]{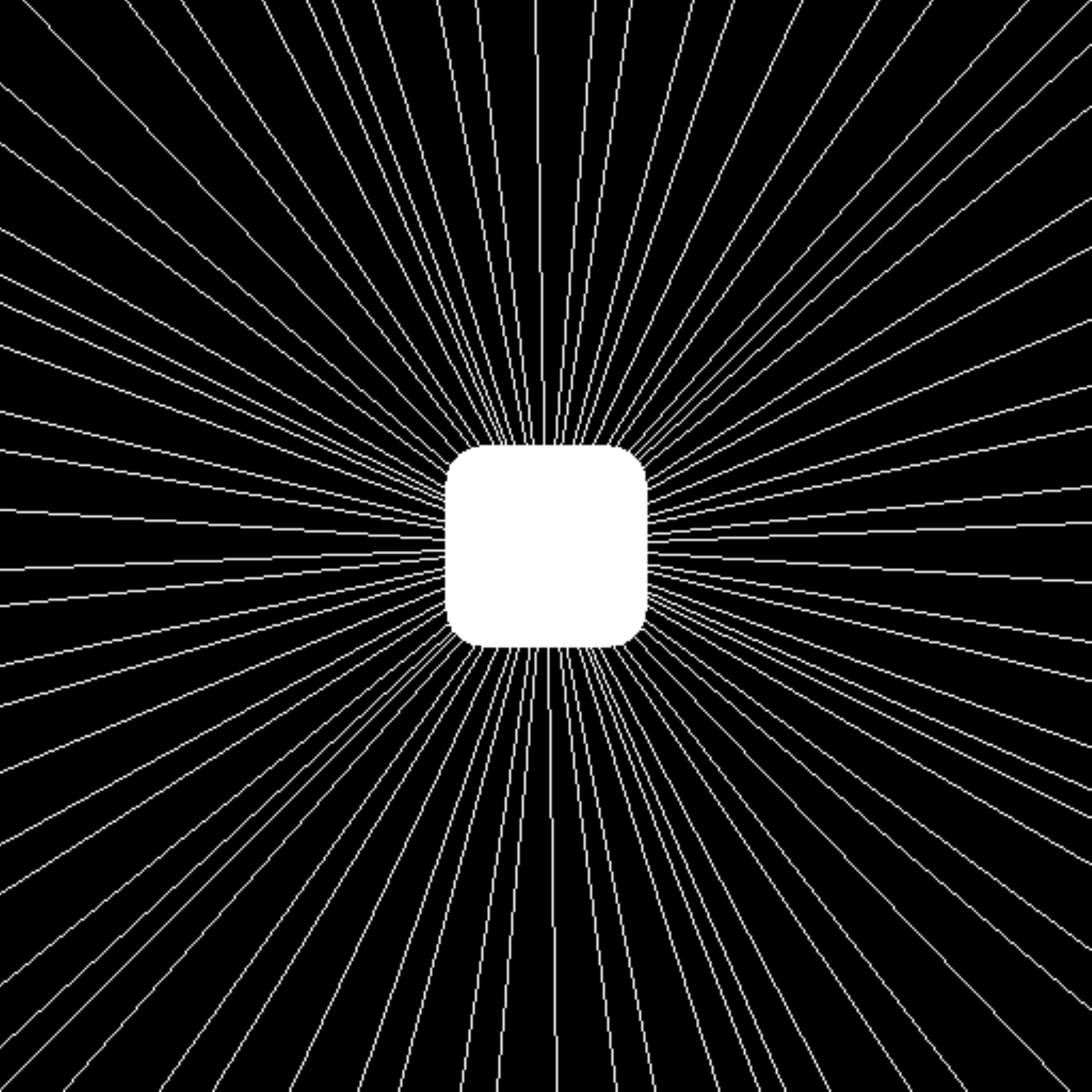} &
\includegraphics[height=6cm]{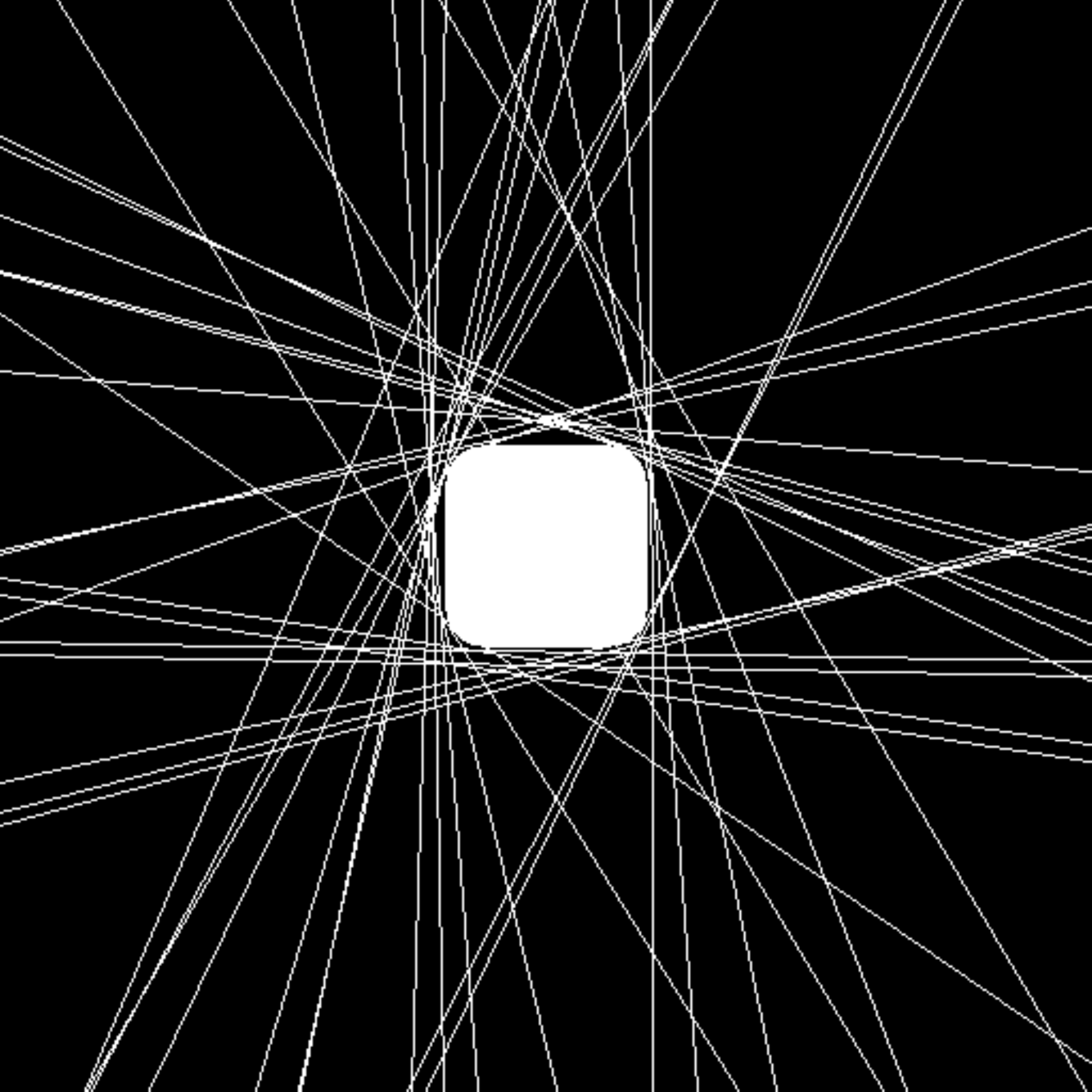}
\\
{\small (a) Golden angle scheme (9.2\%)}&{\small (b) $\pi(p_{\text{rad}})$-based scheme (10\%)}  \\
\includegraphics[height=6cm]{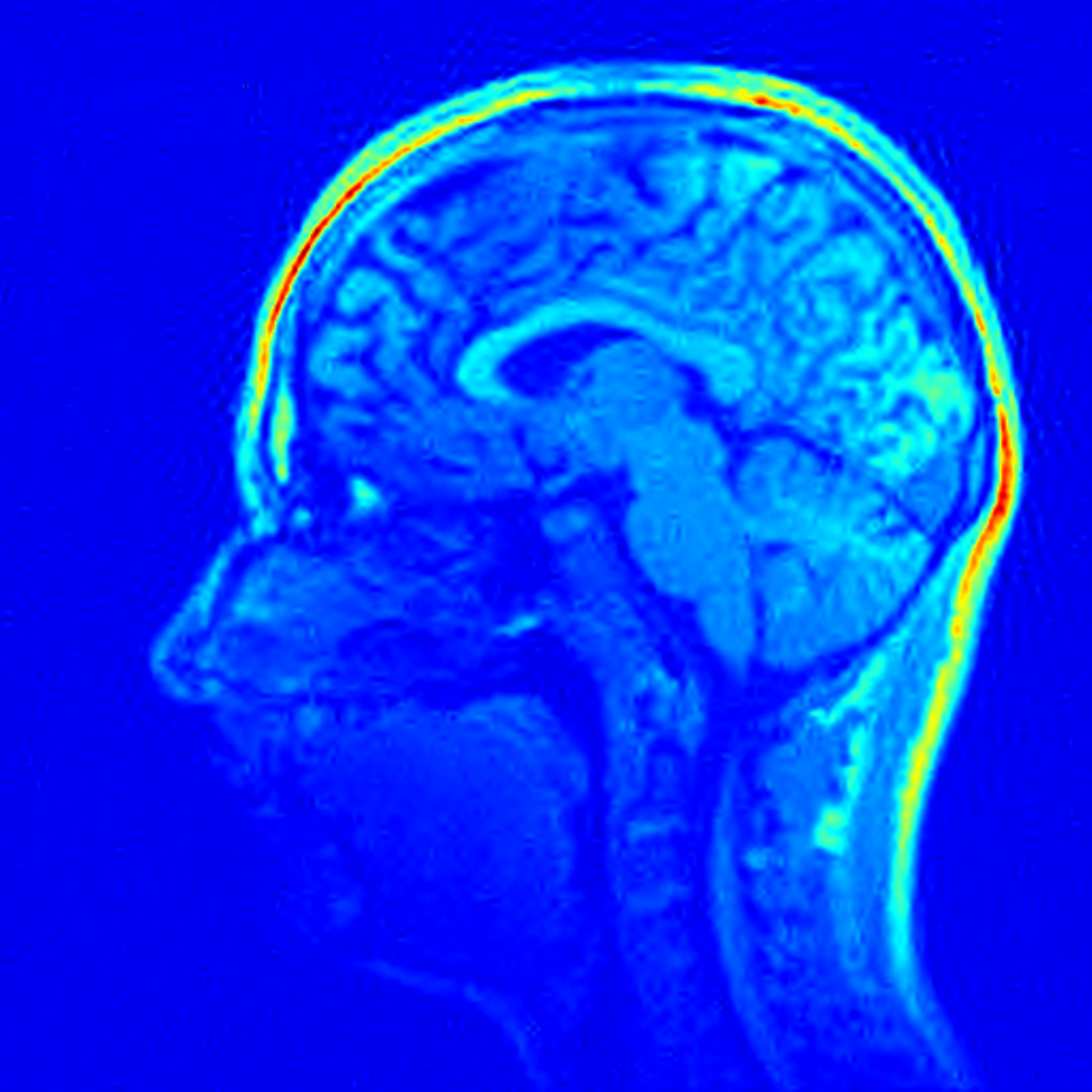}   &
\includegraphics[height=6cm]{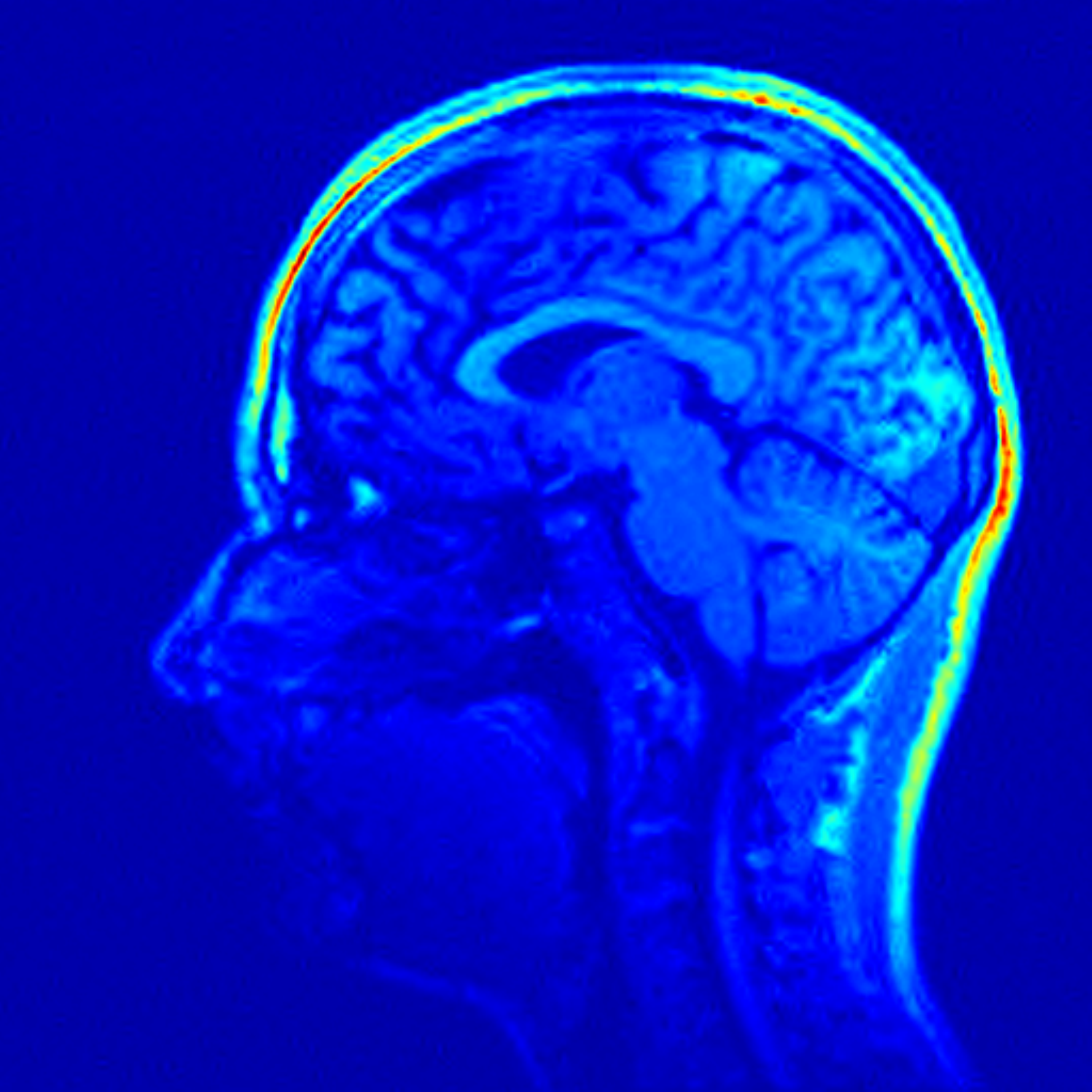} \\
{\small (c) PSNR = 36.34 dB }&{\small (d) PSNR = 38.99 dB} \\
\includegraphics[height=6cm]{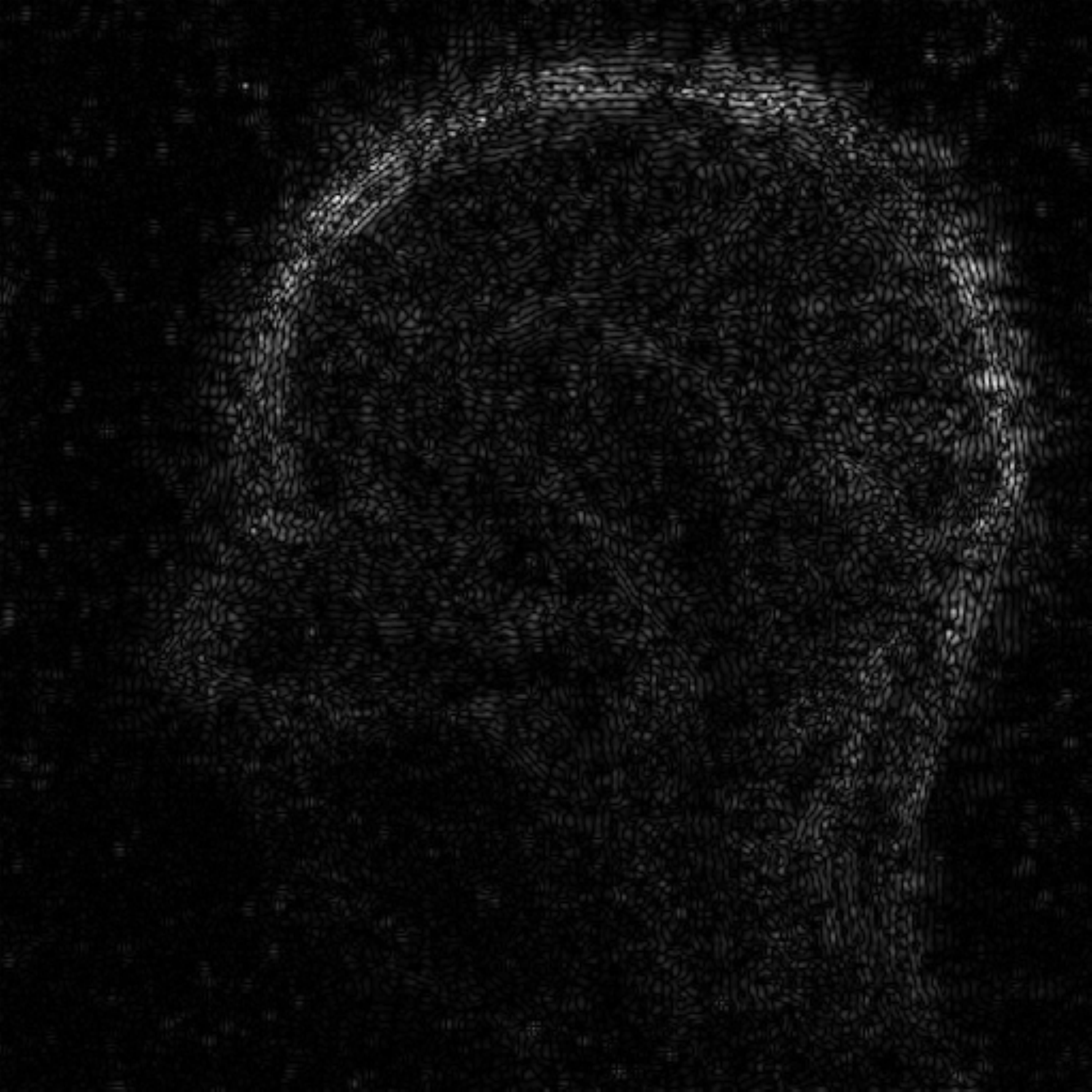} &
\includegraphics[height=6cm]{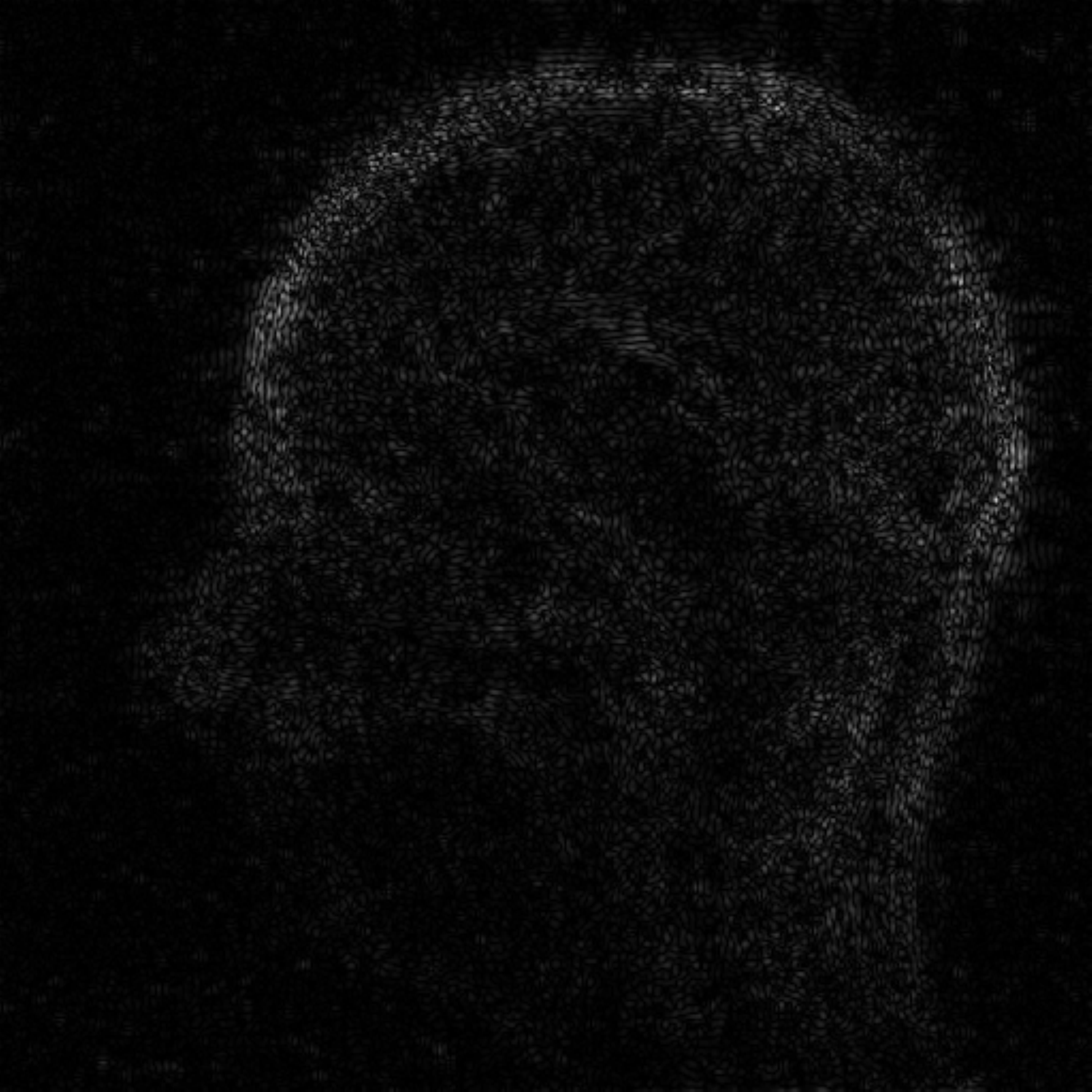} \\
{\small (e) Golden angle }&{\small (f) $\pib(\pb_{\text{rad}})$-based scheme}
\end{tabular}
\caption{\label{fig:nbLines}   Reconstruction examples. We plot schemes made of 37 lines based on the golden angle pattern (a), or based on our method with $\pib(\pb_{\text{rad}})$ (b). Drawing 37 lines in both cases leads to a cover of the sampling space by 9.2\% in the case of the golden angle scheme, and by 10\% for the $\pib(\pb_{\text{rad}})$-based scheme. Note that despite a difference of $0.8\%$ in the covering of the $k$-space, the scanning time is the same for both schemes. In (c) and (d) we display the corresponding reconstructions via $\ell^1$-minimization. We can see that we improve the reconstruction of more than 2 dB with our method. At the bottom, we show the corresponding absolute difference with the reference image. Note that the gray levels have the same scaling in (e) and (f). }
\end{center}
\end{figure}

\paragraph{Remark.} In both settings, for the brain image, schemes based on $\pib\left[\pb_{\text{rad}}\right]$  lead to better reconstruction results in terms of PSNR  than schemes from $\pib\left[\pb_{\text{opt}}\right]$. This can be explained by the fact that $\pb_\text{opt}$ is the probability density given by CS theory which provides guarantees for any $s$-sparse image to reconstruct. However, brain images or natural images have a structured sparsity in the wavelet domain: indeed, their wavelet transform is not uniformly $s$-sparse, the approximation part contains more non-zero coefficients than the rest of the details parts. We can infer that $\pb_\text{rad}$ manages to catch the sparsity structure of the wavelet coefficients of the considered images.

\section{Conclusion}

In this paper, we have focused on constrained acquisition by blocks of measurements. Sampling schemes are constructed by drawing blocks of measurements from a given dictionary of blocks according to a probability distribution $\pib$ that needs to be chosen in an appropriate way. We have presented a novel approach to compute $\pib$ in order to imitate existing sampling schemes in CS that are based on the drawing of isolated measurements. For this purpose, we have defined a notion of dissimilarity measure between a probability distribution on a dictionary of  blocks  and  a probability distribution on a set of isolated measurements. This setting leads to a convex minimization problem in high dimension. In order to compute a solution to this optimization problem, we proposed an efficient numerical approach based on the work of \cite{nesterov2005smooth}. Our numerical study highlights the fact that performing minimization on a metric space rather than a Hilbert space might lead to significant acceleration. Finally, we have presented reconstruction results using this new approach in the case of MRI acquisition. Our method seems to provide better reconstruction results than standard strategies used in this domain. We believe that this last point brings interesting perspectives for 3D MRI reconstruction. 

As an outlook, we plan to extend the proposed numerical method to a wider setting and to provide better theoretical guarantees for cases where the Lipschitz constant of the gradient may vary across the domain. A first step in this direction was proposed recently in \cite{gonzaga2013fine}. We also plan to accelerate the matrix-vectors product involving $\Mb$ by using fast Radon transforms. This step is unavoidable to apply our algorithm in 3D or 3D+t problems for which we expect important benefits compared to the small images we tested until now. Finally we are currently collaborating with physicists at Neurospin, CEA and plan to implement the proposed sampling schemes on real MRI scanners. 

\section*{Acknowledgement}
The authors wish to thank Jean-Baptiste Hiriart-Urruty, Fabrice Gamboa and Alexandre Vignaud for fruitful discussion. 
They also thank the referees for their remarks which helped clarifying the paper.
This work was partially supported by the CIMI (Centre International de Math\'ematiques et d'Informatique)
Excellence program, by ANR SPH-IM-3D (ANR-12-BSV5-0008), by the FMJH Program Gaspard Monge in optimization and operation research, and by the support to this program from EDF.

\appendix

\section{Proof of Proposition \ref{prop:pbDual}}
\label{app:pbDual}
First, we express the Fenchel-Rockafellar dual problem \cite{rockafellar1997convex}:
 \begin{align*}
\min_{\pib \in \Delta_\nbloc} & \| \pb - \Mb \pib \|_{\ell^1} + \alpha \Ec(\pib) \\
&=  \min_{\pib \in \Delta_\nbloc}  \max_{\qb \in B_\infty} \scalprodF{\Mb\pib- \pb}{\qb}   + \alpha \Ec(\pib) \\
&=  \max_{\qb \in B_\infty} \min_{\pib \in \Delta_\nbloc}   \scalprodE{\Mb^*\qb}{\pib}   - \scalprodF{\pb}{ \qb}    + \alpha \Ec(\pib)\\
&= \max_{\qb \in B_\infty} - J_\alpha(\qb) 
\end{align*}
where $B_\infty$ stands for the $\ell^\infty$-ball of unit radius and
\begin{equation}
\label{eq:probjalpha}
J_\alpha(\qb)= - \min_{\pib \in \Delta_\nbloc}   \scalprodE{\Mb^*\qb}{\pib}   - \scalprodF{\pb}{ \qb}    + \alpha \Ec(\pib).
\end{equation}

The solution $\pib(\qb)$ of the minimization problem \eqref{eq:probjalpha} satisfies
\begin{align}
\label{eq:solP}
\Mb^* \qb + \alpha \left( \log (\pib(\qb)) + 1 \right) \in - \Nc_{\Delta_\nbloc} (\pib(\qb))  \quad \text{if} \quad \pib(\qb) \in \text{ri} \left( \Delta_\nbloc \right),
\end{align}
where $ \Nc_{\Delta_\nbloc}(\pib(\qb))$ denotes the normal cone to the set $\Delta_\nbloc$ at the point $\pib(\qb)$, and $\text{ri}\left( \Delta_\nbloc \right)$ denotes the relative interior of $\Delta_\nbloc$.
Equation \eqref{eq:solP} can be rewritten in the following way
\begin{align}
\label{eq:aux1}
  \Mb^* \qb + \alpha  \log(\pib(\qb)) = (-\lambda -\alpha) {\mathds{1}}, \text{with} \quad \lambda \in \Rbb^+ \quad \text{and} \quad \pib(\qb) \in \Delta_\nbloc. 
\end{align}
By choosing $\lambda = \alpha \log \left( \sum_{k=1}^\nbloc \exp\left( -\frac{(\Mb^*\qb)_k}{\alpha} \right) \right) -\alpha$ and plugging it into \eqref{eq:aux1} we get that
\begin{align} 
\label{eq:piFctQ}
\left(\pib(\qb)\right)_j = \frac{\exp\left(  -\frac{(\Mb^*\qb)_j}{\alpha}    \right)}{ \sum_{k=1}^\nbloc  \exp\left(  -\frac{(\Mb^*\qb)_k}{\alpha} \right)}, \qquad \forall j \in \left\lbrace 1, \hdots , \nbloc\right\rbrace.
\end{align}
It remains to plug \eqref{eq:piFctQ} in \eqref{eq:probjalpha} to obtain \eqref{TVdual} with 
$$J_\alpha(\qb) = \scalprodF{\pb}{ \qb}  - \alpha \log \left(   \sum_{k=1}^\nbloc  \exp\left(  -\frac{(\Mb^*\qb)_k}{\alpha} \right) \right).$$

\section{Proof of Proposition \ref{prop:strongconv}}

\label{app:strongconv}
 The neg-entropy is continuous, and twice continuously differentiable on $\text{ri}\left( \Delta_m\right)$. Then, in order to prove its strong convexity, it is sufficient to bound from below its positive diagonal Hessian with respect to $\|\cdot \|_E$. We have
 \begin{align}
 \label{eq:hessEntropie}
 \left\langle \Ec^{''} (\pib ) \hb ,\hb \right\rangle = \sum_{i=1}^m \frac{\left(h_i \right)^2}{\pi_i}, \qquad \text{for} \quad \pib \in \text{ri}\left( \Delta_m\right), \quad \text{and } \quad \hb \in \Rbb^m.
 \end{align}
 Using Cauchy-Schwartz's inequality, 
 \begin{align*}
 \| \hb\|_{\ell^1}^2 &= \left( \sum_{i=1}^m \frac{|h_i|}{\sqrt{\pi_i}} \sqrt{\pi_i}\right)^2 \leq  \left( \sum_{i=1}^m \frac{h_i^2}{\pi_i}\right) \left( \sum_{i=1}^m \pi_i \right) \\
 &\leq \underbrace{\| \pib \|_{\ell^1}}_{=1} \left\langle \Ec^{''} (\pib ) \hb ,\hb \right\rangle.
 \end{align*}
 Therefore, $\Ec$ is $1$-strongly convex on the simplex with respect to $\|.\|_{\ell^1}$. Since for all $p \in \left[1 , \infty \right]$, $\|.\|_{\ell^1} \geq \|.\|_p$, we get:
\begin{align*}
 \| \hb\|_{\ell^p}^2 &\leq  \left\langle \Ec^{''} (\pib ) \hb ,\hb \right\rangle, \quad \forall \hb \in \Rbb^m, \pib \in \text{ri}\left( \Delta_m\right).
 \end{align*}
Moreover if $(\pib_n)_{n\in \mathbb{N}}$ denotes a sequence of $\text{ri}( \Delta_m)$ pointwise converging to the first element of the canonical basis $e_1$ and $h=e_1$, then 
$$
\lim_{n\rightarrow +\infty} \langle \Ec^{''}(\pib_n) h , h\rangle = \|h\|_{\ell^p}^2=1
$$
so that the inequality is tight.

\section{Proof of Proposition \ref{prop:gradLip}}
\label{app:gradLip}

The proof is based on similar arguments as \cite[Theorem 1]{nesterov2005smooth}. 
First, notice that 
\begin{align}
\notag
&\left\langle  \nabla \Ec\left( \pib(\qb_2) \right) - \nabla \Ec\left( \pib(\qb_1) \right) , \pib(\qb_2) - \pib(\qb_1)  \right\rangle \\
\notag
 &= \left\langle  \int_{t=0}^1 \Ec^{''}( \pib_1 + t(\pib_2-\pib_1))(\pib_2 - \pib_1 ) dt  ,   \pib(\qb_2) - \pib(\qb_1)    \right\rangle  \\
 \label{eq:gradEsigmaE}
& \geq \sigma_{\Ec} \left[ \pib_1 , \pib_2\right] \| \pib_2 - \pib_1 \|_E^2,
\end{align}
where $$
\sigma_{\Ec}\left[ \pib_1 , \pib_2 \right] = \inf_{t\in [0,1]} \sigma_{\Ec}(t\pib_1 + (1-t) \pib_2 )
$$  
is the local convexity modulus of $\Ec$ on the segment $\left[ \pib(\qb_1) , \pib(\qb_2) \right] $.

Next, recall that 
$$
J_\alpha(\qb)= \max_{\pib \in \Delta_\nbloc}   \scalprodE{-\Mb^*\qb}{\pib}   + \scalprodF{\pb}{ \qb}    - \alpha \Ec(\pib).
$$
The optimality conditions of the previous maximization problem for $J_\alpha (\qb_1)$ and $J_\alpha (\qb_1)$, $\qb_1, \qb_2 \in F$, lead to
\begin{align*}
&\left\langle -\Mb^* \qb_1 - \alpha \nabla \Ec\left( \pib(\qb_1) \right) , \pib(\qb_2) - \pib(\qb_1) \right\rangle \leq 0 , \\
&\left\langle -\Mb^* \qb_2 - \alpha \nabla \Ec\left( \pib(\qb_2) \right) , \pib(\qb_1) - \pib(\qb_2) \right\rangle \leq 0.
\end{align*}
Combining the two previous inequalities, we can write that for $\qb_1,\qb_2 \in F$:
\begin{align*}
&\left\langle \Mb^* (\qb_1-\qb_2) , \pib(\qb_1) - \pib(\qb_2) \right\rangle \geq   \alpha \left\langle \nabla \Ec\left( \pib(\qb_2) \right) - \nabla \Ec\left( \pib(\qb_1) \right) , \pib(\qb_2) - \pib(\qb_1) \right\rangle, \\
&\Longrightarrow \left\| \Mb^* (\qb_1-\qb_2) \right\|_{E^*} \left\| \pib(\qb_1) - \pib(\qb_2) \right\|_E \stackrel{\eqref{eq:gradEsigmaE}}{\geq} \alpha \sigma_{\Ec}\left[ \pib(\qb_1) , \pib(\qb_2) \right] \left\| \pib(\qb_1) - \pib(\qb_2) \right\|_E^2, \\
& \Longrightarrow \left\| \Mb^* \right\|_{F \rightarrow E^*} \left\|\qb_1-\qb_2\right\|_{F} \left\| \pib(\qb_1) - \pib(\qb_2) \right\|_E \geq \alpha \sigma_{\Ec}\left[ \pib(\qb_1) , \pib(\qb_2) \right] \left\| \pib(\qb_1) - \pib(\qb_2) \right\|_E^2.
\end{align*}

Therefore, we can write that
$$ \left\| \pib(\qb_1) - \pib(\qb_2) \right\|_E  \leq \frac{\left\| \Mb^* \right\|_{F \rightarrow E^*} \left\|\qb_1-\qb_2\right\|_{F} }{\alpha \sigma_{\Ec}\left[ \pib(\qb_1) , \pib(\qb_2) \right]}.
$$
Noting that $\nabla J_\alpha(\qb) = - \Mb \pib(\qb) + \pb$, we can conclude that
\begin{align*}
\|\nabla J_\alpha (\qb_1) - \nabla J_\alpha (\qb_2) \|_{F^*} &= \|\Mb\left( \pib(\qb_1) -\pib(\qb_2) \right) \|_{F^*} \\
&\leq \| \Mb^* \|_{F \rightarrow	 E^*} \|\pib(\qb_1) -\pib(\qb_2)\|_E \\
&\leq \frac{\left\| \Mb^* \right\|^2_{F \rightarrow E^*} }{\alpha \sigma_{\Ec}\left[ \pib(\qb_1) , \pib(\qb_2) \right]}  \left\|\qb_1-\qb_2\right\|_{F} .
\end{align*}
 
Let us consider a sequence $(\qb_n )_{n \in \Nbb}$ converging uniformly towards $\qb \in B_\infty \subset F$. Since $\pib : \qb \in B_\infty \longmapsto \pib(\qb) \in \Delta_m$ is a continuous mapping, $\pib_n \defeqt \pib\left(\qb_n\right)$ converges uniformly towards $\pib(\qb)$. Thus,
$$  \sigma_{\Ec}\left[ \pib(\qb_n) , \pib(\qb) \right] \underset{n \rightarrow +\infty}{\longrightarrow} \sigma_\Ec(\pib(\qb)) .
$$

\section{Proof of Theorem \ref{thm:distanceprimal}}
\label{app:proofBoundPrimal}

Theorem \ref{thm:distanceprimal} is a direct consequence of Lemma \eqref{lem:lemmeprimal} below. 
A similar proof was proposed in \cite{fadili2011total}, but not extended to a general setting.
\begin{lemme}
\label{lem:lemmeprimal}
 Let $f:\Rbb^m\to \Rbb\cup\{+\infty\}$ and $g:\Rbb^n\to \Rbb\cup\{+\infty\}$ denote closed convex functions such that $A\cdot \text{ri}\left( \text{dom}(f)\right) \cap \text{ri}\left( \text{dom}(g)\right) \neq \emptyset$. Assume further that $g$ is $\sigma$-strongly convex w.r.t. an arbitrary norm $\|\cdot\|$. Then
\begin{enumerate}
 \item Function $g^*$ satisfies $\text{dom}(g^*)=\Rbb^n$ and it is differentiable on $\Rbb^n$.
 \item Denote 
\begin{equation}
\label{eq:primal}
 p(x) = f(Ax) + g(x)
\end{equation}
\begin{equation}
\label{eq:dual}
 d(y) = -g^*(-A^*y) -f^*(y)
\end{equation}
and 
\begin{equation*}
 x(y) = \nabla g^*(-A^*y).
\end{equation*}
Let $x^*$ denote the minimizer of \eqref{eq:primal} and $y^*$ denote any minimizer of \eqref{eq:dual}. 
Then for any $y \in \Rbb^m$ we have
\begin{equation}
 \label{eq:guaranteesprimal}
\|x(y) - x^*\|^2 \leq \frac{2}{\sigma} \left( d(y)-d(y^*) \right).
\end{equation}
\end{enumerate}
\end{lemme}

\begin{proof}
Point (i) is a standard result in convex analysis. See e.g. \cite{hiriart1996convex}. 
We did not find the result (ii) in standard textbooks and to our knowledge it is new. 
We assume for simplicity that $g$, $g^*$, $f$ and $f^*$ are differentiable. 
This hypothesis is not necessary and can be avoided at the expense of longer proofs.
First note that 
$$
\inf_{x\in \Rbb^n} p(x) = \sup_{y\in \Rbb^m} -g^*(-A^*y) -f^*(y)
$$ 
by Fenchel-Rockafellar duality. 
Since $g$ is strongly convex $\nabla g$ is a one-to-one mapping and  
\begin{equation}
\label{eq:nablanablastar}
\nabla g(\nabla g^*(x)) =x, \ \forall x\in \Rbb^n.
\end{equation}
The primal-dual relationships read 
\begin{equation*}
\left\{
\begin{array}{ll}
 A^*y^* + \nabla g(x^*) &= 0 \\
 Ax^* - \nabla f^*(y^*) &= 0 
\end{array}\right.
\end{equation*}

So that 
\begin{align*}
x^*&= (\nabla g)^{-1} (-A^*y^*)  \\
&= (\nabla g^*)(-A^*y^*).
\end{align*}

Let us define the following Bregman divergences quantities:
\begin{align*}
D_1(y)&:= f^*(y)-f^*(y^*) - \langle A\nabla g^*(-A^* y^*), y-y^* \rangle\\
D_2(y)&:=g^*(-A^*y) - g^*(-A^*y^*) + \langle A\nabla g^*(-A^*y^*) , y-y^* \rangle.
\end{align*}

By construction
$$
D_1(y)+D_2(y) = d(y)-d(y^*).
$$
Moreover since $y^*$ is the minimizer of $d$ it satisfies $A\nabla g^*(-A^* y^*)=\nabla f^*(y^*)$.
By replacing this expression in $D_1$ and using the fact that $f^*$ is convex we get that
$$
D_1(y)\geq 0, \ \forall y\in \Rbb^n.
$$

Using identity \eqref{eq:nablanablastar} we get:
\begin{align}
\label{eq:D2}
 D_2(y) & = g^*(\nabla g(x(y))) - g^*(\nabla g(x^*)) + \langle x^*, \nabla g(x^*) - \nabla g(x(y)) \rangle.
\end{align}
Moroever, since  (see e.g. \cite{hiriart1996convex})
$$
g(x) + g^*(x^*) = \langle x,x^*\rangle \Leftrightarrow x^*=\nabla g(x),
$$
we get that
\begin{align*}
 g^*(\nabla g(x(y))) &= \langle \nabla g(x(y)), x(y) \rangle - g(x(y)),
\end{align*}
and 
\begin{align*} 
g^*(\nabla g(x^*)) &= \langle \nabla g(x^*), x^* \rangle - g(x^*).
\end{align*}
Replacing these expressions in \eqref{eq:D2} we obtain
\begin{align*}
D_2(y) &= g(x^*) - g(x(y)) + \langle \nabla g(x(y)) , x(y) - x^*\rangle  \\
&\geq \frac{\sigma}{2} \|x(y) - x^*\|^2
\end{align*}
since $g$ is $\sigma$ strongly convex w.r.t $\|\cdot\|$. To sum up we have:
\begin{align*}
 d(y)-d(y^*) & = D_1(y) + D_2(y) \\
&\geq D_2(y) \\
&\geq \frac{\sigma}{2} \|x(y) - x^*\|^2
\end{align*}
which is the desired inequality.
\end{proof}

We now have all the ingredients to prove Theorem \ref{thm:distanceprimal}. 
\begin{proof}[Proof of Theorem \ref{thm:distanceprimal}.]
The proof is a direct consequence of Lemma \ref{lem:lemmeprimal}. It can be obtained by setting $A\equiv \Mb$, $f(y)\equiv \|y - \pb\|_1$ and $g(x) \equiv \alpha \Ec(x)+ \chi_{\Delta_m}(\pib)$, with $\chi_{\Delta_m}$ the indicator function of the set $\Delta_m$. 
Thus $p(x)=f(Ax) + g(x)= F_\alpha(x)$ and $d(y) = J_\alpha(y)$. Then remark that $\pib_k$ defined in Theorem \ref{thm:distanceprimal} satisfies $\pib_k = \nabla g^*(-A^*\yb_k)$. By Proposition \ref{prop:strongconv}, we get that $g$ is $\alpha \sigma_\Ec$-strongly convex w.r.t. $\|\cdot\|_{\ell^p}$, for all $p\in \left[ 1 ; \infty\right]$. It then suffices to use bound \eqref{eq:critCVObj2} together with Lemma \ref{lem:lemmeprimal} to conclude.
\end{proof}

\bibliographystyle{alpha}
\bibliography{mybib}

\end{document}